\documentclass[10pt, twocolumn]{IEEEtran}
\ifCLASSINFOpdf
\else
\fi
\usepackage{amsfonts,amsmath,amsthm,amssymb,graphicx,epstopdf,cite,subfigure}
\usepackage{psfrag,amssymb,amsmath,pifont,cite,graphics,graphicx,epsfig,subfigure,amscd,helvet,multirow,enumerate}
\usepackage{color}
\usepackage{accents}

\newtheorem{theorem}{Theorem}

\newtheorem{proposition}[theorem]{Proposition}

\begin{document}
%
\title{Antenna Grouping based Feedback Compression for FDD-based Massive MIMO Systems}


\author{
Byungju Lee,~\IEEEmembership{Member,~IEEE,}
Junil Choi,~\IEEEmembership{Member,~IEEE,}
Ji-Yun Seol,~\IEEEmembership{Member,~IEEE,}\\
David J. Love,~\IEEEmembership{Fellow,~IEEE,}
and Byonghyo Shim,~\IEEEmembership{Senior Member,~IEEE,}
\thanks{B. Lee and B. Shim are with Institute of New Media and Communications and School of Electrical and Computer Engineering, Seoul National University, Seoul, Korea (e-mail:bjlee$@$islab.snu.ac.kr,bshim$@$snu.ac.kr).

J. Choi is with the University of Texas at Austin, Austin, TX, USA (e-mail:junil.choi@utexaus.edu).

D. J. Love is with School of Electrical and Computer Engineering, Purdue Univ., West Lafayette, IN, USA (email:djlove$@$purdue.edu).

J. Seol is with Samsung Electronics Co., Ltd., Suwon, Korea (email:jiyun.seol$@$samsung.com).}
\thanks{This work was supported by the National Research Foundation of Korea (NRF) grant funded by the Korean government (MSIP) (No. 2014R1A5A1011478 and 2014049151) and MSIP (Ministry of Science, ICT $\&$ Future Planning), Korea in the ICT R$\&$D Program 2013 (No. 1291101110-130010100).}
\thanks{This paper was presented in part at the International Conference on Communications (ICC), 2014 \cite{ByungjuICC} and Vehicular Technology Conference (VTC), 2014 \cite{ByungjuVTC}.}
}

\markboth{To appear in IEEE Transactions on Communications}
{Shell \MakeLowercase{\textit{et al.}}: Bare Demo of IEEEtran.cls for Journals}
%



\maketitle
\begin{abstract}
Recent works on massive multiple-input multiple-output (MIMO) have shown
that a potential breakthrough in capacity gains can be achieved by deploying a very large number of antennas at the basestation.
In order to achieve the performance that massive MIMO systems promise, accurate transmit-side channel state information (CSI) should be available at the basestation.
While transmit-side CSI can be obtained by employing channel reciprocity in time division duplexing (TDD) systems, explicit feedback of CSI from the user terminal to the basestation is needed for frequency division duplexing (FDD) systems.
In this paper, we propose an antenna grouping based feedback reduction technique for FDD-based massive MIMO systems.
The proposed algorithm, dubbed antenna group beamforming (AGB), maps multiple correlated antenna elements to a single representative value using pre-designed patterns.
The proposed method modifies the feedback packet by introducing the concept of a \textit{header} to select a suitable group pattern and a \textit{payload} to quantize the reduced dimension channel vector.
Simulation results show that the proposed method achieves significant feedback overhead reduction over conventional approach performing the vector quantization of whole channel vector under the same target sum rate requirement.
\end{abstract}




%
\IEEEpeerreviewmaketitle

\section{Introduction}
\label{sec:AGB_secI}

Multiple-input multiple-output (MIMO) systems with large-scale transmit antenna arrays, often called massive MIMO, have been of great interest in recent years because of their potential to dramatically improve spectral efficiency of future wireless systems \cite{Marzetta2010,Hoydis2013}.
By employing a simple linear precoding in the downlink and receive filtering in the uplink, massive MIMO systems can control intra-cell interference and thermal noise \cite{Marzetta2010}.
Additionally, massive MIMO can improve the power efficiency by scaling down the transmit power of each terminal inversely proportional to the number of basestation antennas for uplink \cite{Hoydis2013}.

Presently, standardization activity for massive MIMO has been initiated \cite{Standard, Nam2013}, and there is on-going debate regarding the pros and cons of time division duplexing (TDD) and frequency division duplexing (FDD).
In obtaining the channel state information (CSI), FDD requires the CSI to be fed back through the uplink \cite{Lovefeedback} while
no such procedure is required for TDD systems owing to the channel reciprocity \cite{Jose}.
In fact, under the assumption that RF chains are properly calibrated \cite{Emil2013}, the CSI of the downlink can be estimated using the pilot signal in the uplink.
Due to this benefit, most of the massive MIMO works in the literature have focused on TDD \cite{Shepard} (possible exceptions are \cite{Junil2013,NamJSDM,Kuo,CKLove,Truong,Dawei,Davidtraining,Noh2,Joham}).
However, FDD dominates current cellular networks and offers many benefits over TDD (e.g., small latency, continuous channel estimation, backward compatibility), and it is important to identify and develop solutions for potential issues arising from FDD-based massive MIMO techniques.

One well-known problem of FDD system is that the amount of CSI feedback must scale linearly with the number of antennas to control the quantization error \cite{Santipach,Jindal,AuYeung,Mukkavilli}.
Therefore, it is not hard to convince oneself that the overhead of CSI feedback is a serious concern in the massive MIMO regime.
Needless to say, a technique that efficiently reduces the feedback overhead while affecting minimal impact on system performance is crucial to the success of FDD-based massive MIMO systems.

In this paper, we provide a novel framework for FDD-based massive MIMO systems that achieves a reduction in the CSI feedback overhead by exploiting the spatial correlation among antennas.
The proposed algorithm, henceforth dubbed antenna group beamforming (AGB),
maps multiple correlated antenna elements to a single representative value using properly designed grouping patterns.
When the antenna elements are correlated, the loss caused by grouping antenna elements is shown to be small, meaning that grouping of antenna elements with correlated channels is an effective means of reduced dimension channel vector generation.
In fact, by allocating a small portion of the feedback resources to represent the grouping pattern, the number of bits required for channel vector quantization can be reduced substantially, resulting in a significant reduction in feedback overhead.
In order to support the antenna grouping operation, the proposed AGB algorithm uses a new feedback packet structure that divides the feedback resources into two parts: a \textit{header} to indicate the antenna group pattern and a \textit{payload} to indicate the codebook index of the reduced dimension channel vector.
At the user terminal, a pair of group pattern and codeword minimizing the quantization distortion is chosen.
Using the information delivered from the user terminal, the basestation reconstructs the full-dimensional channel vector and then performs transmit beamforming.

In our analysis, we show that when the transmit antenna elements are correlated, the proposed AGB algorithm exhibits smaller quantization distortion than the conventional vector quantization employing a channel statistic-based codebook.
This in turn implies that the number of quantization bits required to meet a certain level of the performance for the AGB algorithm is smaller than that of conventional vector quantization under the same level of quantization distortion.
We also investigate an estimated required number of feedback bits to maintain a constant gap with respect to the system with perfect CSI.
It is shown that the use of antenna grouping in correlated channels enables to considerably reduce the amount of feedback overhead.
Moreover, due to the fact that dimension of the codeword being searched is reduced, and hence the proposed AGB brings additional benefits in search complexity over the conventional vector quantization.
We confirm by simulation on realistic massive MIMO channels that the proposed AGB algorithm achieves up to $20\%$$\sim$$70\%$ savings in feedback information over the conventional vector quantization under the same target sum rate requirement.

The remainder of this paper is organized as follows.
In Section \ref{sec:AGB_secII},  we briefly review the system model and the conventional beamforming technique.
In Section \ref{sec:AGB_secIII}, we provide a detailed description of the proposed AGB algorithm and subspace packing based grouping pattern generation scheme.
We present the simulation results in Section IV and present conclusions in Section V.

\textbf{Notations}: Lower and upper boldface symbols are used to denote vectors and matrices, respectively.
The superscripts $(\cdot)^{H}$, $(\cdot)^{T}$, and $(\cdot)^{\ast}$ denote Hermitian transpose, transpose, and conjugate, respectively.
$\| \mathbf{X}\|$ and $\| \mathbf{X} \|_{F}$ are used as the two-norm and the Frobenius norm of a matrix $\mathbf{X}$, respectively. $E[\cdot]$ denotes the expectation operation, and $\mathcal{CN}(m,\sigma^{2})$ indicates a complex Gaussian distribution with mean $m$ and variance $\sigma^{2}$. $\textrm{tr}(\cdot)$ is the trace operation, and $\textrm{vec}(\mathbf{X}) $ is the vectorization of matrix $\mathbf{X}$. Let $\mathbf{X}_{\Lambda} \in \mathbb{C}^{|\Lambda| \times |\Lambda|}$ denote a submatrix of $\mathbf{X}$ whose $(i,j)$-th entry is $\mathbf{X}\left(\Lambda(i),\Lambda(j)\right)$ for $i,j = 1,\ldots,|\Lambda|$ ($\Lambda$ is the set of partial indices and $|\Lambda|$ is the cardinality of $\Lambda$).

\section{MIMO Beamforming}
\label{sec:AGB_secII}

\subsection{System Model and Conventional Beamforming}
\label{subsec:AGB_secII_subI}

We consider a multiuser multiple-input single-output (MISO) downlink channel with $N_{t}$ antennas at the basestation and $K$ user terminals each with a single antenna\footnote{Note that the proposed method can be easily extended to a MIMO scenario by vectorizing the channel vector corresponding to each receive antenna. For simplicity, we consider the MISO setup for the rest of this paper.} (see Fig. \ref{fig:AGB_Fig1}).
We assume spatially correlated \textit{and} temporally correlated block-fading channels where the channel vector $\mathbf{h}_{i,\ell}$ follows the first-order Gauss-Markov model as
\begin{equation}
\label{eq:start2}
\mathbf{h}_{k,0} = \mathbf{R}_{t,k}^{1/2} \mathbf{g}_{k,0} \nonumber
\end{equation}
\begin{equation}
\label{eq:start3}
\mathbf{h}_{k,\ell} = \eta \mathbf{h}_{k,\ell-1} + \sqrt{1-\eta^{2}} \mathbf{R}_{t,k}^{1/2} \mathbf{g}_{k,\ell}, \,\,\,\,\, \ell \geq 1 \nonumber
\end{equation}
where $\mathbf{R}_{t,k} \in \mathbb{C}^{N_{t} \times N_{t}}$ is the transmit correlation matrix of the $k$-th user \cite{Clerckx}
and $\mathbf{g}_{k,\ell} \in \mathbb{C}^{N_{t}}$ is the innovation process whose elements are independent and identically distributed according to $\mathbf{g}_{k,\ell} \sim \mathcal{CN}(\mathbf{0},\mathbf{I}_{N_{t}})$ and $\eta$ is a temporal correlation coefficient ($0 \leq \eta \leq 1$).
We assume the block-fading channel has a coherence time of $L$, which means that the channel is static for $L$ channel uses in each block and changes from block-to-block.
In this setup, the received signal of the $k$-th user for the $n$-th channel use in the $\ell$-th fading block can be expressed as
\begin{equation}
\label{eq:start}
y_{k,\ell}[n] = \mathbf{h}_{k,\ell}^{H} \mathbf{w}_{k,\ell} s_{k,\ell}[n] + \mathbf{h}_{k,\ell}^{H} \sum_{j \neq k} \mathbf{w}_{j,\ell} s_{j,\ell}[n] + z_{k,\ell}[n]
\end{equation}
where $\mathbf{h}_{k,\ell} \in \mathbb{C}^{N_{t}}$ is the channel vector from the basestation antenna array to the $k$-th user, $\mathbf{w}_{i,\ell} \in \mathbb{C}^{N_{t}}$ is the unit norm beamforming vector ($\|\mathbf{w}_{i,\ell}\|^{2}=1$), $s_{i,\ell}[n] \in \mathbb{C}$ is the message signal for the $i$-th user, and $z_{k,\ell}[n]\sim \mathcal{CN}(0,1)$ is normalized additive white Gaussian noise at the $k$-th user.
%
\begin{figure}[]
\begin{center}
	\includegraphics[width=85mm]{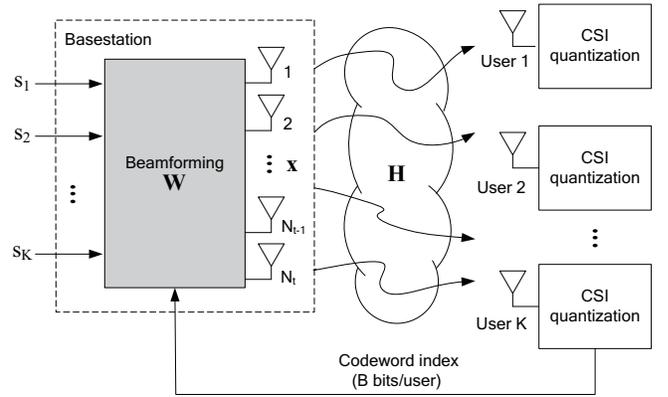}
\caption{CSI feedback in the multi-user downlink system.}
\label{fig:AGB_Fig1}
\end{center}
\vspace{-0.2cm}
\end{figure}
Since the beamforming is performed separately per block, in the sequel we focus on the operation of a single block and drop the fading block index $\ell$.
The matrix-vector form of (\ref{eq:start}) is expressed as
\begin{equation}
\label{eq:total_received}
\mathbf{y}[n] = \mathbf{H} \mathbf{x}[n] + \mathbf{z}[n]
\end{equation}
where $\mathbf{H} = [\mathbf{h}_{1} \,\, \mathbf{h}_{2} \,\, \ldots \,\, \mathbf{h}_{K}]^{H} \in \mathbb{C}^{K \times N_{t}}$ is the composite channel matrix, $\mathbf{z}[n] = [z_{1}[n] \,\,$ $z_{2}[n] \,\, \ldots \,\,$ $z_{K}[n]]^{T} \in \mathbb{C}^{K}$ is the complex Gaussian noise vector ($\mathbf{z}[n] \sim \mathcal{CN}(\mathbf{0},\mathbf{I}_{K})$), $\mathbf{x}[n]$ is the transmit vector normalized with the power constraint ($E[\| \mathbf{x}[n] \|^{2}] = P$), and $\mathbf{y}[n] =[y_{1}[n] \,\, y_{2}[n] \,\, \ldots, \,\, y_{K}[n]]^{T}$ is the vectorized received signal vector.
In order to control the inter-user interference, beamforming is applied using $\mathbf{x}[n] = \mathbf{W} \mathbf{s}[n]$ where $\mathbf{W} = [\mathbf{w}_{1} \,\, \mathbf{w}_{2} \,\, \ldots \,\, \mathbf{w}_{K}] \in \mathbb{C}^{N_{t} \times K}$ and $\mathbf{s}[n] = [s_{1}[n] \,\, s_{2}[n] \,\, \ldots \,\, s_{K}[n]]^{T} \in \mathbb{C}^{K}$ are the beamforming matrix and the message vector, respectively.

In generating the beamforming vectors, we consider zero-forcing beamforming (ZFBF) \cite{Caire2003,Yoo,Jungyong2012,Kailbin,Niranjay} where the right pseudo inverse\footnote{The proposed AGB in this paper is aimed to generate $\hat{\mathbf{h}}_{k}$ with reduced feedback overhead and can be applied to any precoding method including ZFBF.} $\hat{\mathbf{W}}_{\textrm{zf}} = \hat{\mathbf{H}}^{H} \left( \hat{\mathbf{H}} \hat{\mathbf{H}}^{H} \right)^{-1}$ of the quantized channel matrix $\hat{\mathbf{H}} = [\hat{\mathbf{h}}_{1},\hat{\mathbf{h}}_{2},\cdots,\hat{\mathbf{h}}_{K}]^{H}$ is applied to the message vector $\mathbf{s}[n]$ to alleviate the inter-user interference.
In order to satisfy the transmit power constraint, the beamforming vector $\hat{\mathbf{w}}_{k}$ should be normalized as
\begin{equation}
\label{eq:beamforming_vector}
\hat{\mathbf{w}}_{k} = \frac{\hat{\mathbf{W}}_{\textrm{zf}}^{k}}{\| \hat{\mathbf{W}}_{\textrm{zf}}^{k} \|}
\end{equation}
where $\hat{\mathbf{W}}_{\textrm{zf}}^{k}$ is the $k$-th column of $\hat{\mathbf{W}}_{\textrm{zf}}$.
Under the assumption that the basestation allocates equal power for all users\footnote{In this paper, we consider the equal power allocation scenario for simplicity. In order to maximize the sum rate, one might consider more deliberate power allocation strategies (e.g., waterfilling after the block diagonalization \cite{Spencer}).}, the achievable rate of the $k$-th user is
\begin{equation}
\label{eq:sumrate_alluser}
 R_{k} =\log_{2} \left( 1 + \frac{ \frac{P}{K} |\mathbf{h}_{k}^{H} \hat{\mathbf{w}}_{k}|^{2}  }{ 1 + \frac{P}{K} \sum_{j=1, j \neq k}^{K} |\mathbf{h}_{k}^{H} \hat{\mathbf{w}}_{j}|^{2} } \right)
\end{equation}
and the corresponding sum rate becomes $R_{\textrm{sum}} = \displaystyle\sum_{k=1}^{K} R_{k}$.

\begin{figure*}[t]
\centerline{\psfig{figure=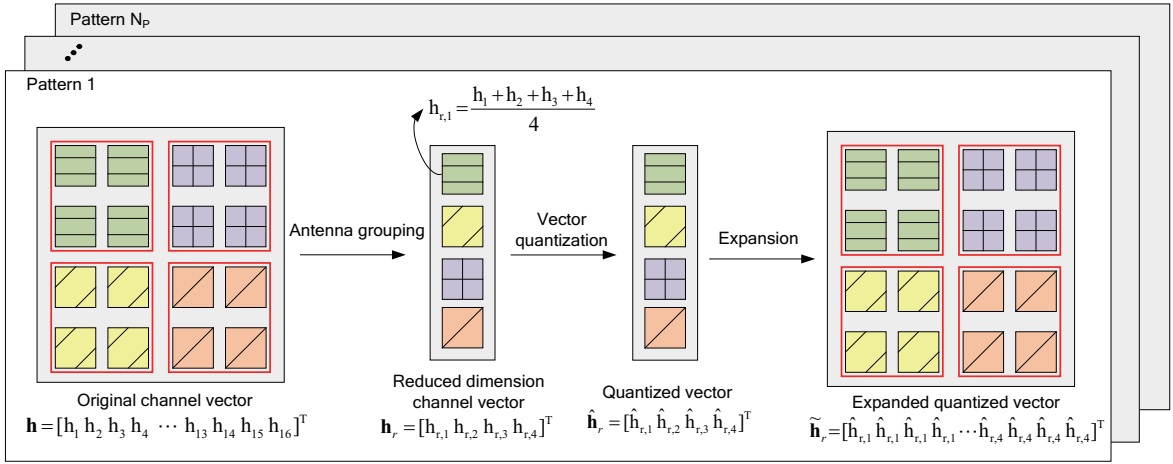, width=155mm}}
\caption{Illustration of the AGB algorithm for $N_{t}=16, N_{g}=4$. The reduced dimension channel vector $\mathbf{h}_{r}$ is obtained by mapping antenna elements of a group as a representative value. Note that $\hat{\mathbf{h}}_{r}$ is the quantized version of $\mathbf{h}_{r}$ and $\tilde{\mathbf{h}}_{r}^{(i)}$ is expanded version of $\hat{\mathbf{h}}_{r}$.} \label{fig:AGB_Fig3}
\vspace{-0.2cm}
\end{figure*}

\subsection{Conventional Limited Feedback}
\label{subsec:AGB_secII_subII}

In order to feed back the CSI, a user quantizes its channel direction $\bar{\mathbf{h}}_{k} = \frac{\mathbf{h}_{k}}{\| \mathbf{h}_{k} \|}$ to a unit norm vector $\hat{\mathbf{h}}_{k}$.
Specifically, the $k$-th user chooses the quantized vector (codeword) $\hat{\mathbf{h}}_{k}$
from a pre-defined $B$-bit codebook set $\mathcal{C}=\{ \mathbf{c}_{1}, \cdots, \mathbf{c}_{2^{B}}\}$ that is closest to its channel direction:\footnote{In practice, each user quantizes the estimated channel, which is obtained using the observations of the pilot signals. Once the channel information corresponding to the pilot signals are estimated, the channel information for the data tones are generated via proper interpolation among pilot channels. With an aim of reducing the pilot overhead of massive MIMO systems, various approaches have been proposed in recent years.
    In \cite{SunhoTsp}, an algorithm exploiting data tones for channel estimation has been proposed and also a pilot allocation strategy based on the sparsity of channel impulse response and compressive sensing (CS) principle \cite{ita_choi15}. Once the estimated channel information is obtained at the receiver, correlation matrix can be estimated using samples of instantaneous channel information (i.e., $\hat{\mathbf{R}}=E[\hat{\mathbf{h}}\hat{\mathbf{h}}^{H}]$).}
\begin{equation}
\label{eq:codebookindex}
\hat{\mathbf{h}}_{k} = \arg \max_{\mathbf{c} \in \mathcal{C} } |\bar{\mathbf{h}}_{k}^{H} \mathbf{c}|^{2}.
\end{equation}
Then, the index of the chosen codeword $\hat{\mathbf{h}}_{k}$ is fed back to the basestation.

In general, the number of bits needed to express the codeword should be scaled with the dimension of the channel vector to be quantized to control the distortion caused by the quantization process.
In particular, when there is no spatial correlation among antenna elements and the random vector quantization (RVQ) codebook is used, the number of feedback bits per user should be scaled with the number of transmit antennas and SNR (in decibels) as \cite{Jindal,Ding}
\begin{equation}
\label{eq:beamforming_vector2}
B_{\textrm{user}}= (N_{t}-1) \log_{2} P \approx \frac{N_{t}-1}{3} P_{\textrm{dB}}
\end{equation}
to maintain a constant gap in terms of the sum rate from the system with perfect CSI.
Hence, when the number of transmit antennas increases, the feedback overhead needs to be increased as well (e.g., $B_{\textrm{user}}=210$ when $N_{t}=64$, $P_{\textrm{dB}}=10$), let alone the computational burden caused by the codebook selection.
Therefore, a reduction in the number of channel vector dimensions would be beneficial in reducing the feedback overhead of the FDD-based massive MIMO systems.

\section{Antenna Grouping based Feedback Reduction}
\label{sec:AGB_secIII}

The key feature of the proposed scheme is to map multiple correlated antenna elements into a single representative value using grouping patterns.
As a result, the channel vector dimension is reduced and a codeword is chosen from the codebook generated by the reduced dimension vector.
When the channel is correlated (i.e., antenna elements in a group are similar), the loss caused by the grouping of antenna elements is negligible and the target performance can be achieved with smaller number of feedback bits than a conventional scheme requires.
In this section, we explain the overall procedure of the proposed AGB algorithm and then discuss the pattern set design problem.
We also analyze the quantization distortion caused by the proposed AGB technique and show that the quantization distortion indeed decreases with the transmit correlation coefficient.

\subsection{AGB Algorithm}
\label{subsec:AGB_secIII_subI}

As mentioned, the AGB algorithm reduces the dimension of the channel vector from $N_{t}$ to $N_{g}$ ($N_{t}>N_{g}$) by mapping multiple correlated antenna elements to a single representative value (see Fig. \ref{fig:AGB_Fig3}).
While a conventional scheme employs all feedback resources ($B$ bits) to express the quantized channel vector, the proposed method uses a part of the feedback resources to quantize the (reduced dimension) channel vector and the rest to express the grouping pattern.
Both basestation and user terminal share a codebook of channel vector and grouping pattern matrices, and thus the receiver feeds back the index of these.
In order to support this operation, we divide the feedback resources into two parts: a header ($B_{p}$ bits) to indicate the antenna group pattern and a payload ($B$-$B_{p}$ bits) to represent an index of the quantized channel vector (see Fig. \ref{fig:AGB_Fig2}). Since antenna grouping based quantization is performed separately for each user, in the sequel we focus on the operation of a single user and drop the user index $k$.

Suppose there are $N_{P}=2^{B_{p}}$ antenna group patterns, then $N_{P}$ distinct reduced dimension channel vectors are generated.
Each pattern converts an $N_{t}$-dimensional channel vector into an $N_{g}$-dimensional vector by multiplying the channel vector by a grouping matrix $\mathbf{G}^{(i)} \in \mathbb{R}^{N_{g} \times N_{t}}$.
The reduced dimension channel vector $\mathbf{h}_{r}^{(i)} \in \mathbb{C}^{N_{g}}$ of the group pattern $i$ is
\begin{equation}
\label{eq:grouping_mapping}
\mathbf{h}_{r}^{(i)} = \mathbf{G}^{(i)} \mathbf{h}, \,\,\,\, i=1,\cdots,N_{P}.
\end{equation}
\begin{figure}[t]
\begin{center}
\subfigure[ ]{
\psfig{file=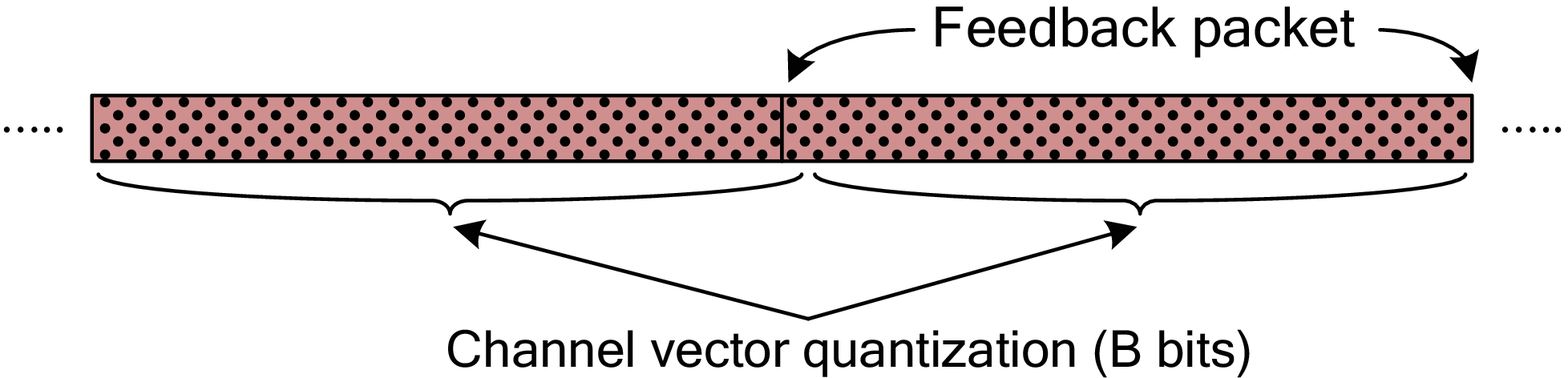,height=24mm,width=70mm} }
\subfigure[ ]{
\psfig{file=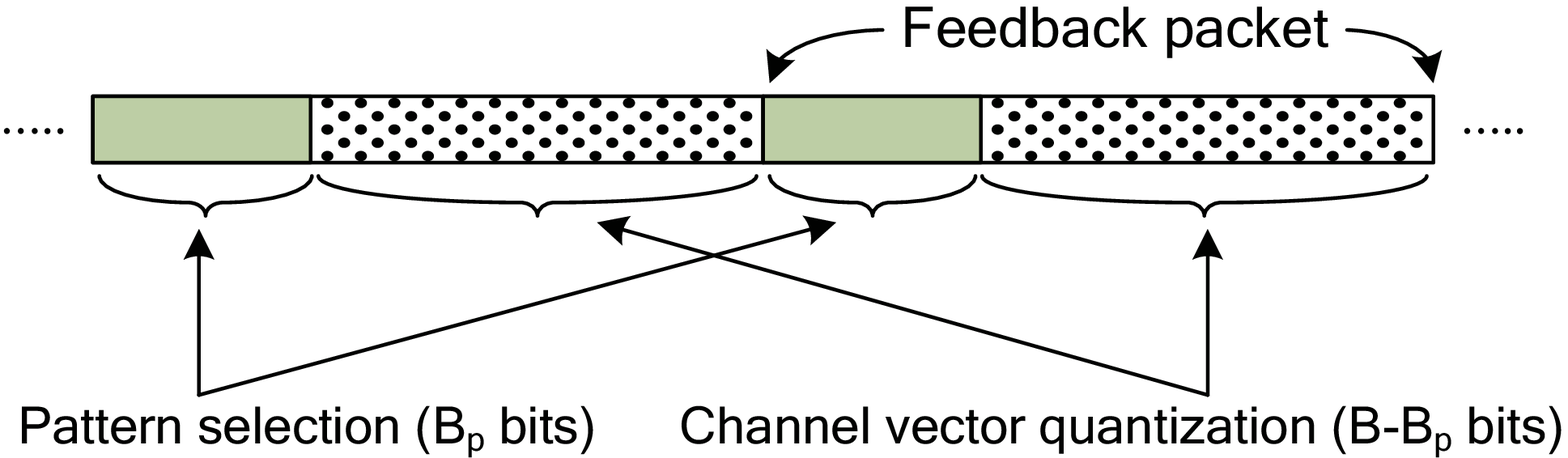,height=24mm,width=72mm} }
\caption{Feedback packet structure: (a) conventional method and (b) proposed method.
}
\label{fig:AGB_Fig2}
\end{center}
\vspace{-0.2cm}
\end{figure}
%
In Fig. \ref{fig:AGB_Fig4}, we illustrate the concept of antenna group patterns.
One simple way to generate the reduced dimension channel vector is to average the channel coefficients in an antenna group.
For example, if $\mathbf{h}=[\textrm{h}_{1} \,\, \textrm{h}_{2} \,\, \textrm{h}_{3} \,\, \textrm{h}_{4}]^{T}$ and two adjacent channel coefficients (first and second, third and fourth) are grouped, the mapping matrix is
\begin{equation} \label{eq:part}
\mathbf{G}^{(i)} = \begin{bmatrix} \frac{1}{2} & \frac{1}{2} & 0 & 0  \\
0 & 0 & \frac{1}{2} & \frac{1}{2}  \end{bmatrix}
\end{equation}
and
$\mathbf{h}_{r}^{(i)}$ $= \mathbf{G}^{(i)}\mathbf{h}=$ $[\frac{\textrm{h}_{1}+\textrm{h}_{2}}{2} \,\,$ $\frac{\textrm{h}_{3}+\textrm{h}_{4}}{2}]^{T}$.

Once the reduced dimension channel vector $\mathbf{h}_{r}^{(i)}$ is obtained,
$\mathbf{h}_{r}^{(i)}$ is quantized by a $B-B_{p}$ bit codebook $\mathcal{C} = \{\mathbf{c}_{1}, \cdots, \mathbf{c}_{2^{B-B{p}}} \}$.
It is worth mentioning that a codebook designed for i.i.d channels is not a proper choice for correlated channels so that we use a channel statistic-based codebook for channel vector quantization \cite{Love} (see Section III.C for details).
The codeword $\hat{\mathbf{h}}_{r}^{(i)}$ maximizing the absolute inner product with $\mathbf{h}_{r}^{(i)}$ is chosen as
\begin{equation}
\label{eq:codebook_index}
\hat{\mathbf{h}}_{r}^{(i)} = \arg \max_{\mathbf{c} \in \mathcal{C}} |\bar{\mathbf{h}}_{r}^{(i) H} \mathbf{c}|^{2}, \,\,\,\, i=1,\cdots,N_{P}
\end{equation}
where $\bar{\mathbf{h}}_{r}^{(i)} = \frac{\mathbf{h}_{r}^{(i)}}{\| \mathbf{h}_{r}^{(i)} \|}$ is the direction of the reduced dimension channel vector for the $i$-th pattern.
This process is repeated for each group pattern and $N_{p}$ candidate codewords $\hat{\mathbf{h}}_{r}^{(i)}$, $i=1,\cdots,N_{P}$, are chosen in total.
%

\begin{figure*}[t]
\centerline{\psfig{figure=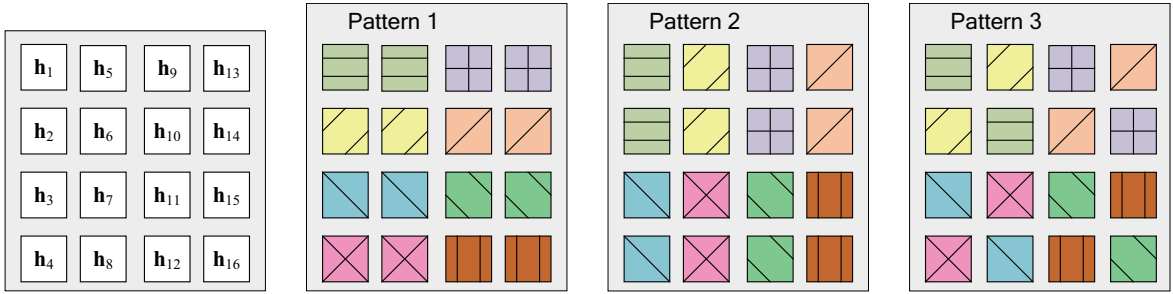, width=155mm}}
\caption{Example of antenna group patterns ($N_{t}=16, N_{g}=8, N_{P}=3$). Antenna elements belonging to the same pattern are mapped to one representative value.} \label{fig:AGB_Fig4}
\vspace{-0.2cm}
\end{figure*}

Once $N_{P}$ candidate codewords are obtained, we need to select the codeword that minimizes the distortion between $\bar{\mathbf{h}}$ and $\hat{\mathbf{h}}_{r}^{(i)}$.
We note that the direct comparison between $\hat{\mathbf{h}}$ and $\hat{\mathbf{h}}_{r}^{(i)}$ is not possible since the dimension of $\hat{\mathbf{h}}_{r}^{(i)} \in \mathbb{C}^{N_{g}}$ is smaller than that of the original channel vector $\mathbf{h} \in \mathbb{C}^{N_{t}}$.
In computing the distortion defined as $D(\mathbf{h},\tilde{\mathbf{h}}_{r}^{(i)}) = E [ \| \mathbf{h} \|^{2} ( 1 - | \bar{\mathbf{h}}^{H} \tilde{\mathbf{h}}_{r}^{(i)} |^{2} ) ] $ caused by the grouping and quantization, therefore, we use $\tilde{\mathbf{h}}_{r}^{(i)} \in \mathbb{C}^{N_{t}}$, an expanded version of $\hat{\mathbf{h}}_{r}^{(i)}$.
The expansion process, which essentially is done by copying each element in $\hat{\mathbf{h}}_{r}^{(i)}$ to $\frac{N_{t}}{N_{g}}$ elements in $\tilde{\mathbf{h}}_{r}^{(i)}$, is performed by multiplying an expansion matrix $\mathbf{E}^{(i)} \in \mathbb{R}^{N_{t} \times N_{g}}$ to $\hat{\mathbf{h}}_{r}^{(i)}$.
The expanded quantized vector $\tilde{\mathbf{h}}_{r}^{(i)}$ is expressed as
\begin{equation}
\label{eq:demapping}
\tilde{\mathbf{h}}_{r}^{(i)} = \mathbf{E}^{(i)}\hat{\mathbf{h}}_{r}^{(i)}, \,\,\,\, i=1,\cdots,N_{P}
\end{equation}
where $\mathbf{E}^{(i)} = \kappa \mathbf{G}^{(i)T}$ and satisfies $\mathbf{G}^{(i)} \mathbf{E}^{(i)} = \mathbf{I}_{N_{g}}$ ($\kappa=\frac{N_{t}}{N_{g}}$).
For example, for the grouping matrix in (\ref{eq:part}),
$\mathbf{E}^{(i)} = \kappa \mathbf{G}^{(i)T} = \begin{bmatrix} 1 & 1 & 0 & 0  \\
0 & 0 & 1 & 1 \end{bmatrix}^{T}$
and the expanded quantized vector is
\begin{equation}
\label{eq:hr_doublehat}
\tilde{\mathbf{h}}_{r}^{(i)} = \mathbf{E}^{(i)} \hat{\mathbf{h}}_{r}^{(i)}= \kappa \mathbf{G}^{(i)T} \hat{\mathbf{h}}_{r}^{(i)} = \left[\hat{\textrm{h}}_{r,1}^{(i)} \,\, \hat{\textrm{h}}_{r,1}^{(i)} \,\,\hat{\textrm{h}}_{r,2}^{(i)} \,\, \hat{\textrm{h}}_{r,2}^{(i)} \right]^{T}.
\end{equation}
The group pattern index $i^{\ast}$ minimizing the distortion between $\mathbf{h}$ and $\tilde{\mathbf{h}}_{r}^{(i)}$ is
\begin{equation}
\label{eq:patter_index}
i^{\ast} = \arg \min_{i=1,\cdots,N_{P}} D(\mathbf{h},\tilde{\mathbf{h}}_{r}^{(i)}).
\end{equation}

Once the pattern index $i^{\ast}$ is obtained, this index and the corresponding codeword index are sent to the basestation.
After receiving the pattern index and codeword index of all user terminals, the basestation decompresses the reduced dimension channel vector via the expansion ($\hat{\mathbf{h}} = \mathbf{E}^{(i^{\ast})} \hat{\mathbf{h}}_{r}^{(i^{\ast})}$) and then performs the beamforming using the composite channel matrix $\hat{\mathbf{H}}$.
A block diagram of the proposed AGB algorithm is depicted in Fig. \ref{fig:AGB_Fig5}.

\subsection{Antenna Group Pattern Generation}
\label{subsec:AGB_secIII_subII}

Since multiple correlated antenna elements are mapped to a single representative value, the AGB algorithm is sensitive to the choice of the antenna group pattern.
Thus, selecting the best pattern among all possible combinations would be an ideal option.
However, since the number of patterns increases exponentially with the number of transmit antennas, it is not possible to investigate all possible patterns for the massive MIMO systems.
Without doubt, a simple yet effective pattern design is crucial to the success of the AGB algorithm.

One easy and intuitive way to construct an antenna group pattern $\mathbf{E}^{(i)}$ is to group highly correlated antenna elements together.
Typically, adjacent antenna elements are highly correlated so that the grouping of nearby antenna elements would be a desirable option in practice (see the example in Fig. \ref{fig:AGB_Fig4}).
Alternatively, one can consider Grassmannian subspace packing in the design of the antenna group patterns \cite{Lovesubspace}.
The main goal of Grassmannian subspace packing is, when the subspace distance metric and the number of feedback bits $B$ are provided, to find a set of $2^{B}$ subspaces in $\mathcal{G}(N_{t},m)$ that maximizes the minimum subspace distance between any pair of subspaces in the set.\footnote{$\mathcal{G}(N_{t},m)$ is the set of $m$-dimensional subspaces in $\mathbb{C}^{N_{t}}$ (or $\mathbb{R}^{N_{t}})$.}
The chordal distance has been popularly used as a metric to measure the distance \cite{love2003}.
Our task of generating the pattern set is similar in spirit to the Grassmannian subspace packing based codebook generation in the sense that we construct a pattern set (containing $2^{B_{p}}$ patterns) from all possible pattern candidates ($\mathbf{E}^{(i)} \in \mathbb{R}^{N_{t} \times N_{g}}$) using a distance metric exploiting the spatial correlation among antenna elements.

In the first step of the pattern set design, we compute the quasi-correlation matrix norm $\|\tilde{\mathbf{R}}_{t}^{(i)} \|_{F}$ to measure the spatial proximity of the antenna elements in the antenna group. The quasi-correlation matrix $\tilde{\mathbf{R}}_{t}^{(i)}$, defined as $\tilde{\mathbf{R}}_{t}^{(i)} = \mathbf{R}_{t}^{1/2} \mathbf{E}^{(i)}$, captures the actual influence of the pattern $\mathbf{E}^{(i)}$ on the transmit correlation matrix $\mathbf{R}_{t}$.
In general, a pattern generated by grouping closely spaced antenna elements tends to have a higher quasi-correlation matrix norm than that generated by grouping antenna elements apart.
Thus, one can deduce that a pattern with a large-correlation matrix norm exhibits lower grouping loss than that with a small quasi-correlation matrix norm.
 For patterns with high quasi-correlation matrix norm, we perform subspace packing to generate $2^{B_{p}}$ patterns (expansion matrices) maximizing the minimum distance metric between any pair of subspaces.
In measuring the distance, we use the correlation matrix distance $d_\textrm{corr}(\mathbf{A}, \mathbf{B})$ between two matrices $\mathbf{A}$ and $\mathbf{B}$ \cite{Herdin}
\begin{equation}
\label{eq:correlation_matrix_distance}
d_\textrm{corr}(\mathbf{A}, \mathbf{B}) = 1 - \frac{\textrm{tr}(\mathbf{A}^{H}\mathbf{B})}{\| \mathbf{A} \|_{\textrm{F}} \| \mathbf{B} \|_{\textrm{F}}}.
\end{equation}
Note that $d_\textrm{corr}(\mathbf{A}, \mathbf{B})$ measures the orthogonality between two correlation matrices $\mathbf{A}$ and $\mathbf{B}$.
 When the correlation matrices are equal up to a scaling factor, $d_{\textrm{corr}}$ is minimized ($d_{\textrm{corr}}=0$). Whereas, when the inner product between the vectorized correlation matrices is zero (i.e., $\textrm{vec}(\mathbf{A})$ and $\textrm{vec}(\mathbf{B})$ are orthogonal), $d_{\textrm{corr}}$ is maximized ($d_{\textrm{corr}}=1$).
 In our numerical simulations, we show that the proposed subspace packing approach achieves a substantial gain over an approach using randomly selected patterns (see Section IV.B).
We summarize the antenna group pattern generation procedures in Table I.

\begin{figure*}[t]
\centerline{\psfig{figure=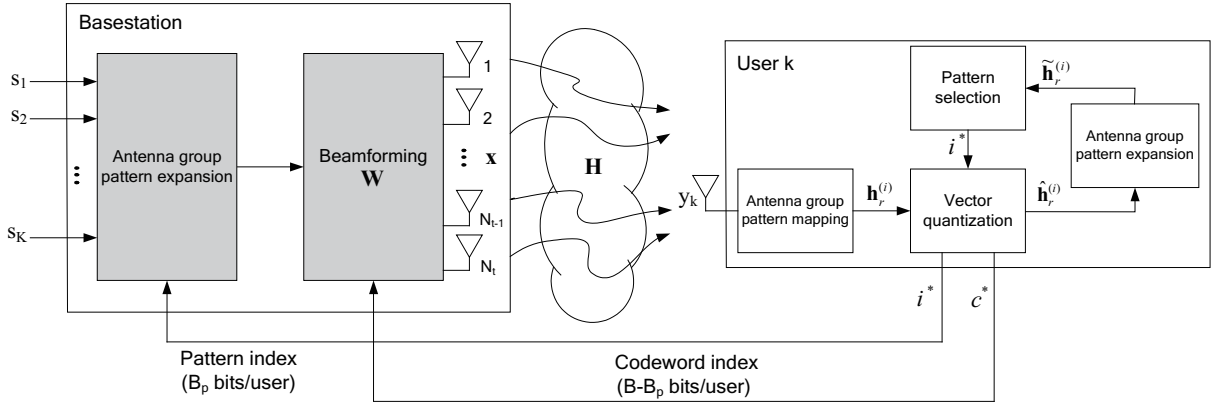, width=160mm}}
\caption{Overall transceiver structure of the proposed AGB technique.
   } \label{fig:AGB_Fig5}
\vspace{-0.2cm}
\end{figure*}
As a further means to lessen the computational burden of pattern generation process, we partition the antenna array into multiple sub-arrays and then apply Grassmannian subspace packing for each sub-array. In the partitioning process, we basically divide the antenna array such that the divided sub-arrays are close to the square matrix. Specifically, when the partitioning level is two, for the $N_{t1} \times N_{t2}$ dimensional antenna array ($N_t = N_{t1} N_{t2}$), we divide the axis with larger dimension. That is, if $N_{t1} \geq N_{t2}$, then the dimension of each sub-array becomes $\frac{N_{t1}}{2} \times N_{t2}$ (see Fig. \ref{fig:AGB_Fig6}). When partitioning level is larger than two, we perform the same procedure for each partitioned sub-array.
In doing so, the number of antenna group pattern candidates is reduced substantially. For example, if $N_{t}=16, N_{g}=8$, and $M=1$ (no partition), then the total number of antenna group pattern candidates $N_{\textrm{max}}$ is about $2 \times 10^{6}$ (see Table I).
Whereas, if the antenna array is partitioned ($M=2$), then $N_{\textrm{max}}$ will be expressed as $N_{\textrm{max}} = \prod_{i=1}^{M} N_{\textrm{max},i}$ where $N_{\textrm{max},i}$ is the total number of candidates for the $i$-th sub-array. Since $N_{\textrm{max},i}=105$ in this case, $N_{\textrm{max}}\approx10^{4}$. Note that since the pattern generation process is performed off the shelf, this process does not affect the real-time operation.

\subsection{Quantization Distortion Analysis}
\label{subsec:AGB_secIII_subIV}
We now turn to the performance analysis of the AGB algorithm.
In our analysis, we analyze the distortion $D$ induced by the quantization of the channel direction vector $\bar{\mathbf{h}} = \frac{\mathbf{h}}{\| \mathbf{h} \|}$,
which is defined as
\begin{align}
\label{eq:distortion}
D &= E \left[ \| \mathbf{h} \|^{2}  - |\mathbf{h}^{H} \tilde{\mathbf{h}}_{r} |^{2} \right] \nonumber \\
&= E \left[ \| \mathbf{h} \|^{2} \left( 1 - | \bar{\mathbf{h}}^{H} \tilde{\mathbf{h}}_{r} |^{2} \right) \right]
\end{align}
where $\tilde{\mathbf{h}}_{r}$ is the expanded version of the quantized vector $\hat{\mathbf{h}}_{r}$ (see (\ref{eq:demapping})).
\begin{figure}[t]
\centerline{\psfig{figure=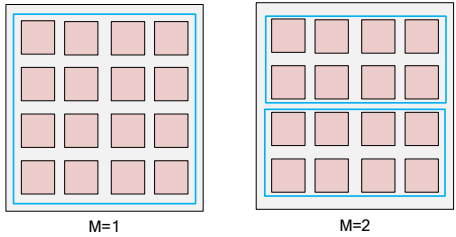, width=60mm}}
\caption{Antenna group pattern of the partitioned sub-arrays when $M=1$ (no partition) and $M=2$.} \label{fig:AGB_Fig6}
\vspace{-0.2cm}
\end{figure}

In the evaluation of the distortion $D$, we use the quantization cell upper bound (QUB) \cite{YooFeedback}. As mentioned, since a codebook designed for the i.i.d channels is not the right choice for correlated channels, we employ a channel statistic-based codebook obtained by applying the transmit correlation matrix $\mathbf{R}_{t}^{1/2}$ to the codebook generated from the Grassmannian line packing.
Let $\mathbf{f}_{i} \in \mathbb{C}^{r}$ be the $i$-th unit norm vector generated from the Grassmannian line packing, then the set of $B$-bit codewords for the channel statistic-based codebook is \cite{Love}
\begin{equation}
\label{eq:CDITcodebook}
\mathcal{C} = \left\{\mathbf{c}_{1},\cdots,\mathbf{c}_{2^{B}}\right\}=\left\{ \frac{\mathbf{R}_{t}^{1/2} \mathbf{f}_{1}}{\|\mathbf{R}_{t}^{1/2} \mathbf{f}_{1}\|},\cdots,\frac{\mathbf{R}_{t}^{1/2} \mathbf{f}_{2^{B}}}{\|\mathbf{R}_{t}^{1/2} \mathbf{f}_{2^{B}}\|}\right\}.
\end{equation}
When the channel statistic-based codebook is used, the normalized distortion $\frac{D}{E[\|\mathbf{h}\|^{2}]}=1-|\bar{\mathbf{h}}^{H} \mathbf{c}_{i}|^{2}$ between the channel direction vector $\bar{\mathbf{h}}$ and codebook vector $\mathbf{c}_{i}$ can be upper bounded as \cite{Clerckx}
\begin{eqnarray}
\label{eq:quantization_cell_Yoo}
\mathcal{R}_{i} \approx \{ \bar{\mathbf{h}}: 1-|\bar{\mathbf{h}}^{H} \mathbf{c}_{i}|^{2} \leq \delta\}
\end{eqnarray}
where $\delta = \frac{\sigma_{2}^{2}}{\sigma_{1}^{2}} 2^{-\frac{B}{r-1}}$ $(\sigma_{i}$ is the $i$-th largest singular value of the transmit correlation matrix $\mathbf{R}_{t} \in \mathbb{C}^{r \times r})$ and $B$ is the number of quantization bits.

\begin{table*}
\begin{center}
\caption{Summary of the antenna group pattern generation.}  
\begin{tabular}{ll} \hline
  Initialization      & $B_{p}$: the number of bits for the pattern set \\
                      & $\mathcal{S}$: the set of patterns to be selected \\   \hline
 Main operation & 1) Initialize the index set $\Omega=\{1,\ldots,N_{\textrm{max}}\}$ where \\
 & \hspace{1.2cm} $N_{\textrm{max}} = \frac{\prod_{n=0}^{N_{g}-1} {N_{t} - n \kappa \choose \kappa}}{N_{g}!}$. \\
           & \hspace{0.3cm} $\kappa = \frac{N_{t}}{N_{g}}$ is the number of elements in an antenna group. \\
           & 2) For each pattern $i \in \Omega$, calculate the Frobenius norm of the quasi-correlation matrix \\
           & \hspace{1.2cm}  $r_{i} = \| \tilde{\mathbf{R}}_{t}^{(i)} \|_{F}$. \\
           & \hspace{0.3cm} Without loss of generality, assume $r_{1} \geq r_{2} \geq \cdots \geq r_{N_{\textrm{max}}}$. \\
           & 3) Choose $J(\geq 2^{B_{p}})$ patterns $\mathcal{T} = \{r_{1},\ldots,r_{J}\}$. \\
           & 4) Apply the subspace packing to $\mathcal{T}$ to generate the pattern set $\mathcal{S}$. \\
           & \hspace{0.3cm} Construct $N_{c} = {J \choose 2^{B_{p}}}$ candidate sets $\{\mathcal{S}_{k}\}_{k=1}^{N_{c}}$ where $\mathcal{S}_{k}=\{\tilde{\mathbf{R}}_{t}^{k,1},\tilde{\mathbf{R}}_{t}^{k,2},\cdots,\tilde{\mathbf{R}}_{t}^{k,2^{B_{p}}} \}.$\\
           & \hspace{0.3cm} $\tilde{\mathbf{R}}_{t}^{k,i}$ is the $i$-th quasi-correlation matrix from the $k$-th candidate set. \\
           & 5) Calculate the minimum $d_{\textrm{corr}}$ of a $\mathcal{S}_{k}$ \\
           & \hspace{1.2cm} $d_{k,\textrm{min}}(\mathcal{S}_{k}) =  \min_{1 \leq m \leq n \leq 2^{B_{p}}} d_{\textrm{corr}} (\tilde{\mathbf{R}}_{t}^{k,m},\tilde{\mathbf{R}}_{t}^{k,n})$. \\
           & \hspace{0.3cm} Decide the pattern set $\mathcal{S}$ \\
           & \hspace{1.2cm} $\mathcal{S} = \arg \max_{k=1,\cdots,N_{c}} d_{k,\textrm{min}}(\mathcal{S}_{k})$. \\ \hline
\end{tabular}
\end{center}
\vspace{-0.2cm}
\label{tb:AGB_PatternTable}
\end{table*}
In our analysis, we restrict our attention to the scenario where two antenna elements are mapped to a single representative value for mathematical tractability.
Nevertheless, since the key factor affecting the quantization distortion is the transmit correlation coefficient (see (\ref{eq:proposition_approximation})), our results can be readily applied to the general scenario where more than two antenna elements are grouped together.
The minimal set of assumptions used for the analytical tractability are as follows:
\begin{description}
\item[A-i)] The channel vector $\bar{\mathbf{h}}$ is partitioned into two subvectors $\bar{\mathbf{h}}_{\textrm{A}}$ and $\bar{\mathbf{h}}_{\textrm{B}}$.
    $\bar{\mathbf{h}}_{\textrm{A}}$ and $\bar{\mathbf{h}}_{\textrm{B}}$ are composed of odd and even entries in $\bar{\mathbf{h}}$ (i.e., $\bar{\mathbf{h}}_{\textrm{A}} = [\mathbf{h}_{1} \,\, \mathbf{h}_{3} \,\, \cdots]^{T}$, $\bar{\mathbf{h}}_{\textrm{B}} = [\mathbf{h}_{2} \,\, \mathbf{h}_{4} \,\, \cdots]^{T}$). Thus, $N_{g} = \frac{N_{t}}{2}$.
\item[A-ii)] The reduced dimension channel vector $\mathbf{h}_{r}$ is designed such that $\mathbf{h}_{r}=\bar{\mathbf{h}}_{\textrm{A}}$.
For example, if $\bar{\mathbf{h}} = [\bar{\textrm{h}}_{1} \,\, \bar{\textrm{h}}_{2} \,\, \bar{\textrm{h}}_{3} \,\, \bar{\textrm{h}}_{4}]^{T}$, then the antenna grouping pattern is
\begin{equation} \label{eq:part_pattern}
\mathbf{G} = \begin{bmatrix} 1 & 0 & 0 & 0  \\
0 & 0 & 1 & 0  \end{bmatrix}
\end{equation}
and hence $\mathbf{h}_{r} = \bar{\mathbf{h}}_{\textrm{A}}=[\bar{\textrm{h}}_{1} \,\, \bar{\textrm{h}}_{3}]^{T}$ ($\bar{\mathbf{h}}_{\textrm{B}} = [\bar{\textrm{h}}_{2} \,\, \bar{\textrm{h}}_{4}]^{T}$).
\item[A-iii)] Antenna elements in a group are highly correlated. That is, $E[|\bar{\mathbf{h}}_{\textrm{A}}^{H} \hat{\mathbf{h}}_{r}|^{2}] \approx E[|\bar{\mathbf{h}}_{\textrm{B}}^{H} \hat{\mathbf{h}}_{r}|^{2}]$ where $\hat{\mathbf{h}}_{r}$ is the quantized vector of $\mathbf{h}_{r}$ and generated from the channel statistic-based codebook. Note that this assumption justifies the use of antenna grouping pattern in (\ref{eq:part_pattern}).
\end{description}
It is worth mentioning that depending on the way of grouping elements of antenna array, there are several ways to generate subvectors $\bar{\mathbf{h}}_{\textrm{A}}$ and $\bar{\mathbf{h}}_{\textrm{B}}$. Also, the group pattern used in (17) might be worse than the pattern generated by the subspace packing or proposed in (8). Thus, assumptions in A-i) and A-ii) would be clearly pessimistic, but makes our analysis tractable.
The following proposition provides an approximate upper bound of the quantization distortion $D$ under these assumptions.
\begin{proposition}
\label{lemma1}
The quantization distortion $D$ of the AGB algorithm under the channel statistic-based codebook satisfies
\begin{eqnarray}
\label{eq:proposition}
D \lesssim  N_{t} \delta + \xi \sqrt{2 N_{t} \delta}
\end{eqnarray}
%
where $\delta= \frac{\sigma_{2}^{2}}{\sigma_{1}^{2}} 2^{-\frac{B-B_{p}}{N_{g}-1}}$ is an upper bound of the normalized distortion between $\bar{\mathbf{h}}_{\textrm{A}}$ and $\hat{\mathbf{h}}_{r}$ as defined in (\ref{eq:quantization_cell_Yoo}) ($\sigma_{i}$ is the $i$-th largest singular value of $\mathbf{R}_{t,\textrm{A}}$)\footnote{For example, if $\mathbf{R}_{t} = \begin{bmatrix} 1 & \rho & \rho^{2} & \rho^{3}  \\ \rho & 1 & \rho & \rho^{2} \\ \rho^{2} & \rho & 1 & \rho \\ \rho^{3} & \rho^{2} & \rho & 1 \end{bmatrix}$ and $\textrm{A}=\{1,3\}$, then $\mathbf{R}_{t,\textrm{A}} = \begin{bmatrix} 1 & \rho^{2}  \\ \rho^{2} & 1 \end{bmatrix}$.} and $\xi$ is the correlation coefficient between two random variables $\| \mathbf{h}\|^{2}$ and $1 - | \bar{\mathbf{h}}^{H} \tilde{\mathbf{h}}_{r} |^{2}$.
\end{proposition}
\begin{proof}
Using (\ref{eq:distortion}), we have
\begin{align}
\label{eq:distortion_upperbound}
D &= E \left[ \| \mathbf{h} \|^{2} \left( 1 - | \bar{\mathbf{h}}^{H} \tilde{\mathbf{h}}_{r} |^{2} \right) \right] \\
\label{eq:distortion_upperboundsub}
&= E \left[ \| \mathbf{h} \|^{2} \right] \left( 1 - E \left[| \bar{\mathbf{h}}^{H} \tilde{\mathbf{h}}_{r} |^{2} \right] \right) + Cov(\| \mathbf{h} \|^{2},1 - | \bar{\mathbf{h}}^{H} \tilde{\mathbf{h}}_{r} |^{2}) \\
\label{eq:distortion_upperboundsub2}
&= E \left[ \| \mathbf{h} \|^{2} \right] \left( 1 - E \left[| \bar{\mathbf{h}}^{H} \tilde{\mathbf{h}}_{r} |^{2} \right] \right)  \nonumber \\ \hspace{0.1cm} &+ \xi \sqrt{ Var \left(\| \mathbf{h} \|^{2} \right)}  \sqrt{ Var \left( 1 - | \bar{\mathbf{h}}^{H} \tilde{\mathbf{h}}_{r} |^{2} \right)} \\
\label{eq:distortion_upperboundsub3}
&\leq E \left[ \| \mathbf{h} \|^{2} \right] \left( 1 - E \left[| \bar{\mathbf{h}}^{H} \tilde{\mathbf{h}}_{r} |^{2}\right] \right)  \nonumber \\ \hspace{0.1cm} &+ \xi \sqrt{ Var \left[\| \mathbf{h} \|^{2} \right]}  \sqrt{ 1 - \left(E \left[| \bar{\mathbf{h}}^{H} \tilde{\mathbf{h}}_{r} |^{2} \right] \right)^{2}}
\end{align}
where (\ref{eq:distortion_upperboundsub}) is because $Cov(X,Y) = E[XY]-E[X]E[Y]$, (\ref{eq:distortion_upperboundsub2}) is because $Cov(X,Y) = \xi \sqrt{Var(X)} \sqrt{Var(Y)}$
and (\ref{eq:distortion_upperboundsub3}) is because $Var \left( 1 - | \bar{\mathbf{h}}^{H} \tilde{\mathbf{h}}_{r} |^{2} \right) = E \left[|\bar{\mathbf{h}}^{H} \tilde{\mathbf{h}}_{r}|^{4}\right] - \left(E \left[| \bar{\mathbf{h}}^{H} \tilde{\mathbf{h}}_{r} |^{2} \right] \right)^{2}  \leq 1 - \left(E \left[| \bar{\mathbf{h}}^{H} \tilde{\mathbf{h}}_{r} |^{2} \right] \right)^{2}$.

The normalized distortion term $1 - E \left[| \bar{\mathbf{h}}^{H} \tilde{\mathbf{h}}_{r} |^{2}\right]$ in the right-hand side of (\ref{eq:distortion_upperboundsub3}) is approximately upper bounded as
\begin{align}
\label{eq:distortion_change}
  & 1 - E \left[ |\bar{\mathbf{h}}^{H} \tilde{\mathbf{h}}_{r}|^{2} \right]  \nonumber \\ &\stackrel{(a)}{=} 1 - E \left[ |\bar{\mathbf{h}}_{\textrm{A}}^{H} \hat{\mathbf{h}}_{r} + \bar{\mathbf{h}}_{\textrm{B}}^{H} \hat{\mathbf{h}}_{r}|^{2} \right]  \nonumber \\
  & = 1 - E [ |\bar{\mathbf{h}}_{\textrm{A}}^{H} \hat{\mathbf{h}}_{r}|^{2} + |\bar{\mathbf{h}}_{\textrm{B}}^{H} \hat{\mathbf{h}}_{r}|^{2} + (\bar{\mathbf{h}}_{\textrm{A}}^{H} \hat{\mathbf{h}}_{r})^{\ast}(\bar{\mathbf{h}}_{\textrm{B}}^{H} \hat{\mathbf{h}}_{r}) \nonumber \\  &\,\,\,\,\,\, + (\bar{\mathbf{h}}_{\textrm{B}}^{H} \hat{\mathbf{h}}_{r})^{\ast}(\bar{\mathbf{h}}_{\textrm{A}}^{H} \hat{\mathbf{h}}_{r}) ] \nonumber \\
& = 1 - E \left[ |\bar{\mathbf{h}}_{\textrm{A}}^{H} \hat{\mathbf{h}}_{r}|^{2} \right] - E \left[ |\bar{\mathbf{h}}_{\textrm{B}}^{H} \hat{\mathbf{h}}_{r}|^{2} \right] \nonumber \\ &\,\,\,\,\,\, - 2 E \left[ \textrm{Re} (\bar{\mathbf{h}}_{\textrm{A}}^{H} \hat{\mathbf{h}}_{r})^{\ast}(\bar{\mathbf{h}}_{\textrm{B}}^{H} \hat{\mathbf{h}}_{r})\right] \nonumber \\
  &\stackrel{(b)}{\approx} 1 - E \left[ |\bar{\mathbf{h}}_{\textrm{A}}^{H} \hat{\mathbf{h}}_{r}|^{2} \right] - E \left[ |\bar{\mathbf{h}}_{\textrm{B}}^{H} \hat{\mathbf{h}}_{r}|^{2} \right] \nonumber \\
  &\stackrel{(c)}{=} 1 - 2 E \left[ |\bar{\mathbf{h}}_{\textrm{A}}^{H} \hat{\mathbf{h}}_{r}|^{2} \right]   \nonumber \\
   &\stackrel{(d)}{\lesssim}  1 -  2 (1-\delta ) E \left[ \|  \bar{\mathbf{h}}_{\textrm{A}}\|^{2} \right] \nonumber \\
   &\stackrel{(e)}{=} \delta
\end{align}
where (a) is because $|\bar{\mathbf{h}}^{H} \tilde{\mathbf{h}}_{r}|^{2} = |\bar{\mathbf{h}}_{\textrm{A}}^{H} \hat{\mathbf{h}}_{r} + \bar{\mathbf{h}}_{\textrm{B}}^{H} \hat{\mathbf{h}}_{r}|^{2}$, (b) follows from $E [ \textrm{Re} ((\bar{\mathbf{h}}_{\textrm{A}}^{H} \hat{\mathbf{h}}_{r})^{\ast}(\bar{\mathbf{h}}_{\textrm{B}}^{H} \hat{\mathbf{h}}_{r}))] \approx 0$ (see Appendix A),
(c) follows from A-iii), and (d) follows from the QUB in (\ref{eq:quantization_cell_Yoo}).
That is, by plugging $\bar{\mathbf{h}} = \frac{\bar{\mathbf{h}}_{\textrm{A}}}{\| \bar{\mathbf{h}}_{\textrm{A}} \|}$, $\mathcal{C} = \left\{\frac{\mathbf{R}_{t,\textrm{A}}^{1/2} \mathbf{f}_{1}}{\|\mathbf{R}_{t,\textrm{A}}^{1/2} \mathbf{f}_{1}\|},\cdots,\frac{\mathbf{R}_{t,\textrm{A}}^{1/2} \mathbf{f}_{2^{B-B_{p}}}}{\|\mathbf{R}_{t,\textrm{A}}^{1/2} \mathbf{f}_{2^{B-B_{p}}}\|}\right\}$, and $\delta= \frac{\sigma_{2}^{2}}{\sigma_{1}^{2}} 2^{-\frac{B-B_{p}}{N_{g}-1}}$ into (\ref{eq:quantization_cell_Yoo}), we get $E\left[|\bar{\mathbf{h}}_{\textrm{A}}^{H} \hat{\mathbf{h}}_{r}|^{2}\right] \geq (1-\delta) E\left[\| \bar{\mathbf{h}}_{\textrm{A}} \|^{2}\right]$.
Finally, (e) follows from $E \left[ \|  \bar{\mathbf{h}}_{\textrm{A}}\|^{2} \right] = \frac{N_{g}}{N_{t}} = \frac{1}{2}$.

Plugging (\ref{eq:distortion_change}) into (\ref{eq:distortion_upperboundsub3}), we have
\begin{align}
\label{eq:distortion_change3}
D &\leq E \left[ \| \mathbf{h} \|^{2} \right] \delta  + \xi \sqrt{ Var \left[\| \mathbf{h} \|^{2} \right]}  \sqrt{ 1-(1-\delta)^{2} }  \\
\label{eq:distortion_change32}
&= N_{t} \delta  + \xi \sqrt{N_{t} \left( 1-(1-\delta)^{2}  \right)} \\
\label{eq:distortion_change33}
&\approx  N_{t} \delta  + \xi \sqrt{2 N_{t} \delta}
\end{align}
where (\ref{eq:distortion_change32}) is because $E[\| \mathbf{h} \|^{2}] = N_{t}$ and $Var[\| \mathbf{h} \|^{2}] = N_{t}$, and (\ref{eq:distortion_change33}) is because $\delta(2-\delta) \approx 2\delta$ where $\delta \ll 2$, which is the desired result.
\end{proof}

We note that the relationship between the quantization distortion $D$ and the transmit antenna correlation is not clearly shown in (\ref{eq:proposition}).
When a specific correlation model is used, however, we can observe the relationship between the two.
For example, if the exponential correlation model is employed, the transmit correlation matrix $\mathbf{R}_{t}$ is expressed as \cite{Martin}
\begin{equation} \label{eq:exponential_matrixfirst}
\mathbf{R}_{t} = \begin{bmatrix} 1 & \rho & \cdots & \rho^{N_{t}-1}  \\ \rho^{\ast} & 1 & \cdots & \rho^{N_{t}-2} \\
\vdots & \vdots & \ddots & \vdots \\
(\rho^{\ast})^{N_{t}-1} & (\rho^{\ast})^{N_{t}-2} & \cdots & 1
\end{bmatrix}
\end{equation}
where $\rho = \alpha e^{j \theta}$ is the transmit correlation coefficient, and $\alpha$ is the magnitude of correlation coefficient, and $\theta$ is the phase of the coefficient.
When the number of transmit antennas $N_{t}$ is large, (non-ordered) singular value $\mu_{i}$ of $\mathbf{R}_{t}$ approximately behaves as\cite{Gray}
\begin{align}
\label{eq:eigenvalue_approx}
\mu_{i} &\approx \sum_{k= -(N_{t}-1)}^{N_{t}-1} \rho^{|k|} e^{j \frac{2 \pi i k}{N_{t}}} \nonumber \\
&\approx \frac{1 - \rho^{2}}{ 1 + \rho^{2} - 2 \rho \cos (\frac{2 \pi i}{N_{t}})}, \,\,\, i=1,\ldots,N_{t}.
\end{align}
Using the first and second largest singular values of\footnote{Due to the symmetric property of $\mu_{i}$ (i.e., $\mu_{N_{t}}>\mu_{N_{t}-1}=\mu_{1}>\mu_{N_{t}-2}=\mu_{2}>\cdots$), $\sigma_{1} = \mu_{N_{t}}$ and $\sigma_{2}=\mu_{N_{t}-1}=\mu_{1}$.} (\ref{eq:eigenvalue_approx}), we have
\begin{align}
\label{eq:singular_ratio}
\frac{\sigma_{2}}{\sigma_{1}}  \approx \frac{1 + \rho^{2} - 2 \rho}{1 + \rho^{2} - 2 \rho \cos(2 \pi \frac{N_{t}-1}{N_{t}})}.
\end{align}
Using this together with $\delta= \frac{\sigma_{2}^{2}}{\sigma_{1}^{2}} 2^{-\frac{B-B_{p}}{N_{g}-1}}$ in Proposition 3.1, we have
\begin{eqnarray}
\label{eq:proposition_approximation}
D &\lesssim  N_{t} \frac{(1 + \rho^{2} - 2 \rho)^2}{\left(1 + \rho^{2} - 2 \rho \cos(2 \pi \frac{N_{t}-1}{N_{t}})\right)^2} 2^{-\frac{B-B_{p}}{N_{g}-1}} \nonumber \\ \hspace{0.1cm} &+ \xi \sqrt{ 2N_{t}} \frac{1 + \rho^{2} - 2 \rho}{1 + \rho^{2} - 2 \rho \cos(2 \pi \frac{N_{t}-1}{N_{t}})} 2^{-\frac{B-B_{p}}{2(N_{g}-1)}}.
\end{eqnarray}
%
In (\ref{eq:proposition_approximation}), we observe that the quantization distortion $D$ decreases with the correlation coefficient $\rho$.
 Fig. \ref{fig:AGB_DistortionAnalysis} plots the normalized quantization distortion $\frac{D}{E[\|\mathbf{h}\|^{2}]}$ as a function of the correlation coefficient $\rho$. Note that $\delta = \frac{\sigma_{2}^{2}}{\sigma_{1}^{2}} 2^{-\frac{B}{r-1}}$ is obtained from the assumption that all non-zero singular values except the dominant one (i.e., $\sigma_{1}$) are the same ($\sigma_{2} = \sigma_{3} = \cdots$). Note also that $\mathcal{R}_{i}$ is tight in a regime where transmit antennas are highly correlated since $\frac{\sigma_{2}}{\sigma_{1}}$ (i.e., $\delta$) decreases with the correlation coefficient. Readers are referred to \cite{Mukkavilli,Clerckx} for more details.  We observe that if $|\rho|>0.3$, the quantization distortion $D$ of the AGB algorithm is better (smaller) than that of conventional vector quantization. We can also observe that the analysis matches well with the simulation results when the transmit antennas are highly correlated ($|\rho|>0.6$). However, when the magnitude of $\rho$ is small, the assumption in A-iii) is violated so that the proposed bound is invalid.
 \begin{figure}
\centering
\centerline{\epsfig{figure=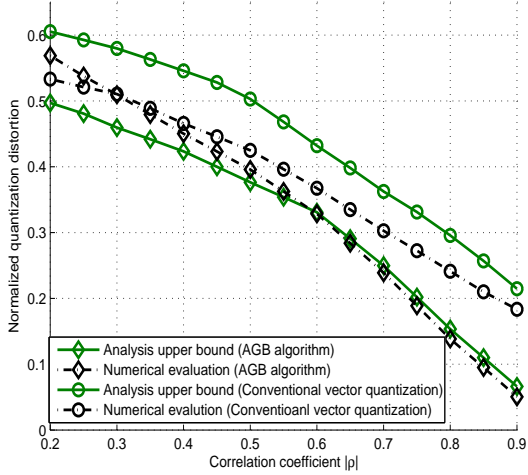, height = 70mm, width=80mm}}
\caption{Normalized quantization distortion as a function of the correlation coefficient ($N_{t}=16, N_{g}=8, B=16, B_{p}=8)$.} \label{fig:AGB_DistortionAnalysis}
\vspace{-0.2cm}
\end{figure}

The following proposition provides the upper bound of the sum rate gap $\Delta R(P)$.
\begin{proposition}
When an equal power allocation per user is applied,
the sum rate gap (per user) between the ZFBF with perfect CSI and the proposed method satisfies
\begin{align}
\label{eq:corollary}
\Delta R(P) \lesssim  \log_{2}\left(1+P \frac{K-1}{K} ( N_{t} \delta  + \xi \sqrt{2 N_{t} \delta}  )\right)
\end{align}
where $\Delta R(P)$ is the difference between the achievable rate achieved by (\ref{eq:sumrate_alluser}) and $\mathbf{w}_{k} = \frac{\mathbf{W}_{\textrm{zf}}^{k}}{\| \mathbf{W}_{\textrm{zf}}^{k}\|}$ where $\mathbf{W}_{\textrm{zf}} = \mathbf{H}^{H}\left(\mathbf{H}\mathbf{H}^{H}\right)^{-1}$.
\end{proposition}
\begin{proof}
Note that $\Delta R(P)$ is given by  $\Delta R(P)= E [ \log_{2} (1 + \frac{P}{K} |\mathbf{h}_{k}^{H} \mathbf{w}_{k}|^{2} ) ] - E [\log_{2} ( 1 + \frac{ \frac{P}{K} |\mathbf{h}_{k}^{H} \hat{\mathbf{w}}_{k}|^{2}  }{ 1 + \frac{P}{K} \sum_{j=1, j \neq k}^{K} |\mathbf{h}_{k}^{H} \hat{\mathbf{w}}_{j}|^{2} } )]$. Using Jensen's inequality, $\Delta R(P)$ can be upper bounded as \cite{Jindal}
\begin{align}
\label{eq:rategap}
 \Delta R(P) &\leq E \left[ \log_{2} \left( 1 + \frac{P}{K} \sum_{j=1, j\neq k}^{K} | \mathbf{h}_{k}^{H} \hat{\mathbf{w}}_{j}|^{2} \right) \right] \nonumber \\
 &\leq  \log_{2} \left( 1 + \frac{P}{K} E \left[\sum_{j=1, j\neq k}^{K} | \mathbf{h}_{k}^{H} \hat{\mathbf{w}}_{j}|^{2} \right] \right).
\end{align}
Using orthogonality between $\hat{\mathbf{w}}_{j}$ and $\tilde{\mathbf{h}}_{r,k}$,
\begin{align}
\label{eq:orthogonality}
 \| \mathbf{h}_{k}\|^{2} \geq | \mathbf{h}_{k}^{H} \hat{\mathbf{w}}_{j}|^{2}  + \| \mathbf{h}_{k}\|^{2} | \bar{\mathbf{h}}_{k}^{H} \tilde{\mathbf{h}}_{r,k}|^{2},
 \end{align}
then (\ref{eq:rategap}) becomes
\begin{align}
\label{eq:rategap2}
 \Delta R(P) &\leq  \log_{2} \left( 1 + \frac{P}{K} E \left[\sum_{j=1, j\neq k}^{K} \| \mathbf{h}_{k}\|^{2} (1 - |\bar{\mathbf{h}}_{k}^{H} \tilde{\mathbf{h}}_{r,k}|^{2}) \right] \right) \nonumber \\
 &=  \log_{2} \left( 1 + P\frac{(K-1)}{K}  D \right)
\end{align}
 where (\ref{eq:rategap2}) is due to $D = E[\| \mathbf{h}_{k}\|^{2} (1 - |\bar{\mathbf{h}}_{k}^{H} \tilde{\mathbf{h}}_{r,k}|^{2})]$.
Using (\ref{eq:proposition}) and (\ref{eq:rategap2}), we get the desired result.
\end{proof}
Next proposition specifies the number of feedback bits needed to maintain a constant rate gap from the system with perfect CSI.
\begin{proposition}
In order to maintain a rate gap (between the ZFBF with perfect CSI and the proposed method) within $\log_{2} \beta$ bps/Hz per user, it is sufficient to scale the number of bits per user according to
\begin{align}
\label{eq:corollary2}
B &\approx B_{p} + (N_{g}-1) [\log_{2} \left( \frac{(1 + \rho^{2} - 2 \rho)^2}{(1 + \rho^{2} - 2 \rho \cos(2 \pi \frac{N_{t}-1}{N_{t}}))^2} \right) \nonumber \\ \hspace{0.1cm} &- 2\log_{2}\left(\frac{-\xi+\sqrt{\xi^{2}+4(\beta-1)\frac{K}{P(K-1)}}}{2\sqrt{N_{t}}}\right) ].
\end{align}
\end{proposition}
\begin{proof}
In order to maintain a rate loss of $\Delta R(P) \leq \log_{2}\beta$ bps/Hz per user, we set the rate gap upper bound given in Proposition 3.2 equal to the maximum allowable gap of $\log_{2} \beta$ as
\begin{align}
\Delta R(P) \lesssim  \log_{2}\left(1+P \frac{K-1}{K} ( N_{t} \delta  + \xi \sqrt{2 N_{t} \delta }  )\right) \triangleq \log_{2} \beta.
\end{align}
By inverting (36) and solving for $B$, we get the desired result.
\end{proof}
Fig. \ref{fig:AGB_DistortionAnalysis2} plots the sum rate as a function of $\textrm{SNR}$ when $B$ in (\ref{eq:corollary2}) is applied. We fix $\beta=2$ in order to maintain a SNR gap of 3 dB.
We observe that by using a proper scaling of feedback bits, we can limit the rate loss within 3 dB (in fact around 2 dB), as desired.

\begin{figure}
\centering
\centerline{\epsfig{figure=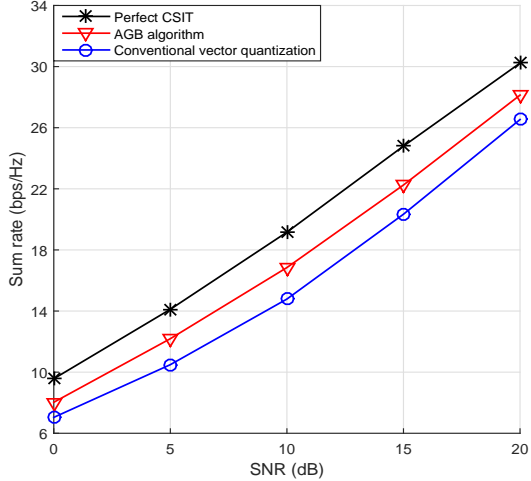, height = 70mm, width=80mm}}
\caption{Sum rate as a function of $\textrm{SNR}$ when ($N_{t}=32, N_{g}=16, K=4, \rho=0.9, \xi=0.05, B_{p}=8)$.} \label{fig:AGB_DistortionAnalysis2}
\vspace{-0.2cm}
\end{figure}

\subsection{Comments on Complexity}

In this subsection, we discuss the complexity of the AGB algorithm and conventional vector quantization.
While the major operation of the conventional approach is to search the codeword index, computations associated with pattern index selection, grouping and expansion process, pattern generation are additionally required for the proposed method.
We first analyze the computational complexity of the pattern index selection, grouping process, and the expansion process, which are performed on the fly.
Denoting the complexity associated with pattern index selection, grouping process, and the expansion process as $C_{\textrm{p}}$, $C_{\textrm{g}}$, $C_{\textrm{e}}$, respectively, then the number of required floating-point operations (flops) for each step is as follows \cite{Golub}:
\begin{itemize}
\item $C_{\textrm{p}}$ requires $4N_{t}$ flops for computing the distortion $D$ in (\ref{eq:patter_index}).
\item $C_{\textrm{g}}$ requires $(2N_{t}-1)N_{g}$ flops for the matrix multiplication in (\ref{eq:grouping_mapping}).
\item $C_{\textrm{e}}$ requires $(2N_{g}-1)N_{t}$ flops for the matrix multiplication in (\ref{eq:demapping}).
\end{itemize}
Note that these operations need to be computed for $N_{P}$ times.
We next measure the computational complexity of the pattern generation process. The number of required flops for computing the quasi-correlation matrix norm $\|\tilde{\mathbf{R}}_{t}^{(i)} \|_{F}$ and the correlation matrix distance $d_{\textrm{corr}}$ can be obtained as $2N_{t}N_{g}$ and $2N_{g}^{2}N_{t}+4N_{t}N_{g}$, respectively.
Then, according to Table I, the total computational complexity becomes $2N_{\textrm{max}}N_{t}N_{g} + {J \choose N_{P}}{N_{P} \choose 2}(2N_{g}^{2}N_{t}+4N_{t}N_{g})$. Note that $N_{\textrm{max}}$ can be reduced significantly by applying the partitioning approach discussed in Section III.B.
As mentioned, the pattern generation process does not affect the real-time operation since this process is performed off the shelf.
In contrast to the operations we just described, the codeword search complexity is quantified by $O(\cdot)$ notation. Note that codeword search complexity grows exponentially with the dimension of the vector to be quantized and the codeword search complexity for the conventional approach and proposed method is given by $O\left(N_{t}2^{B}\right)$ and $N_{P}O\left(N_{g} 2^{B-B_{p}}\right)$, respectively. Note also that the complexity of additional operations (i.e., $N_{P}(C_{\textrm{p}} + C_{\textrm{g}} + C_{\textrm{e}})$) is much smaller than that of the codeword search complexity $N_{P}O\left(N_{g} 2^{B-B_{p}}\right)$. Overall, the proposed method brings additional benefits in search complexity over the conventional approach due to the fact that dimension of the codeword being searched is reduced.

\section{Simulation Results and Discussions}
\label{sec:AGB_secIV}

\subsection{Simulation Setup}
\label{subsec:AGB_secIV_subI}

In this section, we compare the sum rate performance of the conventional vector quantization using the channel statistic-based codebook \cite{Love} and the proposed AGB algorithm. While all the feedback resources ($B$-bit) are used to quantize the channel vector $\mathbf{h}_{k}$ in the conventional vector quantization approach, $B$-bit feedback resource is divided into $B_{q}$ (channel vector quantization) and $B_{p}$ (pattern selection) in the proposed method. To express the feedback allocation, we use the notation $B=(B_{q},B_{p})$ in the sequel.
As a pattern set, we use the combination of patterns for each sub-array.
Let $B_{p,\textrm{sub}}$ be the number of pattern bits of each sub-array ($B_{p,\textrm{sub}}=\frac{B_{p}}{M}$), then $2^{B_{p,\textrm{sub}}}$ patterns are generated by applying the proposed subspace packing approach.
As a transmit antenna model, we consider the exponential correlation model in (\ref{eq:exponential_matrixfirst}) \cite{Adhikary} and two-dimensional uniform planar array (UPA) model \cite{Noh}. We use Jakes' model \cite{Proakis} for the temporal correlation coefficient $\eta = J_{0} (2 \pi f_{D} \tau)$ where $J_{0}(\cdot)$ is the $0$-th order Bessel function of the first kind, $f_{D} = v f_{c}/c$ denotes the maximum Doppler frequency, and $\tau=5ms$ is the channel instantiation interval.
With the user speed $v=3km/h$, the carrier frequency $f_{c}=2.5GHz$, and the speed of light $c=3 \times 10^8 m/s$, the temporal correlation coefficient becomes $\eta=0.9881$.
Assuming a $5$ms coherence time and frame structure of 3GPP LTE FDD systems \cite{Standard}, each fading block consists of $L\approx10$ static channel uses.

\begin{figure}
\centering
\centerline{\epsfig{figure=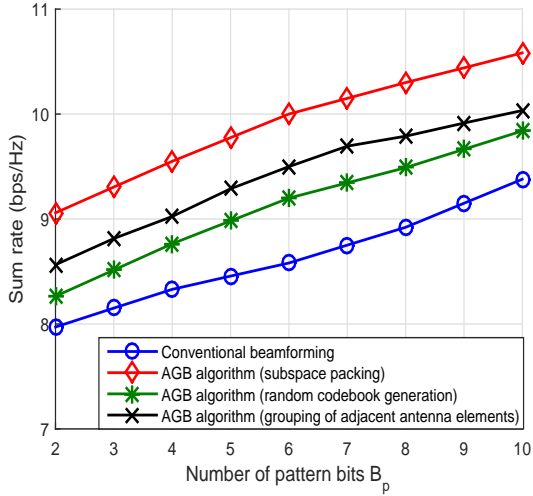, height = 70mm, width=80mm}}
\caption{Sum rate as a function of the number of pattern bits ($N_{t}=16, M=2, K=1, N_{g}=8, B_{q}=16, B=B_{q}+B_{p}), SNR$=10$ dB$.} \label{fig:AGB_Fig11}
\vspace{-0.2cm}
\end{figure}

\begin{figure}
\centering
\centerline{\epsfig{figure=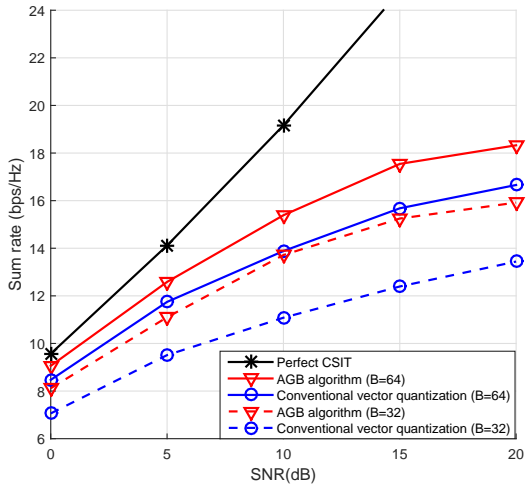, height = 70mm, width=80mm}}
\caption{Sum rate as a function of SNR when $B=(B_{q},8)$ ($N_{t}=32, M=4, K=4, N_{g}=16,\alpha=0.8$).} \label{fig:AGB_Fig9a}
\vspace{-0.2cm}
\end{figure}

\subsection{Simulation Results}
\label{subsec:AGB_secIV_subII}

We first consider the exponential channel model
\begin{equation}
\label{eq:exponential_correlation}
r_{ij} = \left\{
\begin{array}{ll}
\rho_{k}^{|j-i|} & \quad i \leq j \\
(\rho_{k}^{|j-i|})^H & \quad i > j \end{array} \right.
\end{equation}
where $r_{ij}$ is the $(i,j)$-th element of $\mathbf{R}_{t,k}$ and $\rho_{k} = \alpha e^{j \theta_{k}}$ is a transmit correlation coefficient for the $k$-th user where $\alpha$ is the magnitude of correlation coefficient and $\theta_{k}$ is the phase of the $k$-th user.
Note that the phase of each user is randomly generated from $-\pi$ to $\pi$ and independent among each user.
Note also that all users have the same transmit correlation coefficient $|\rho_{k}| = \alpha$ since $\alpha$ is determined by the antenna spacing at the basestation.

In order to observe the effectiveness of the subspace packing approach discussed in Section III.B, we compare the proposed approach to the random pattern generation and grouping of adjacent antenna elements.
In our simulations, we set $B_{q}=16, N_{t}=16, N_{g}=8, M=2, K=1$ and measure the sum rate as a function of the number of pattern bits $B_{p}$.
To set the same level of feedback, we set $B = B_{q} + B_{p}$ bit for the conventional vector quantization.
Overall, we observe from Fig. \ref{fig:AGB_Fig11} that the subspace packing approach provides a considerable sum rate gain over the approach using randomly generated patterns, AGB with grouping of adjacent antenna elements as well as the conventional vector quantization technique.
For example, to achieve $9$ bps/hz, AGB with subspace packing requires $B=18$ bits while AGB with grouping of adjacent antenna elements, AGB with random patterns and conventional vector quantization require $20$, $21$ and $24$ bits, respectively.

\begin{figure}
\centering
\centerline{\epsfig{figure=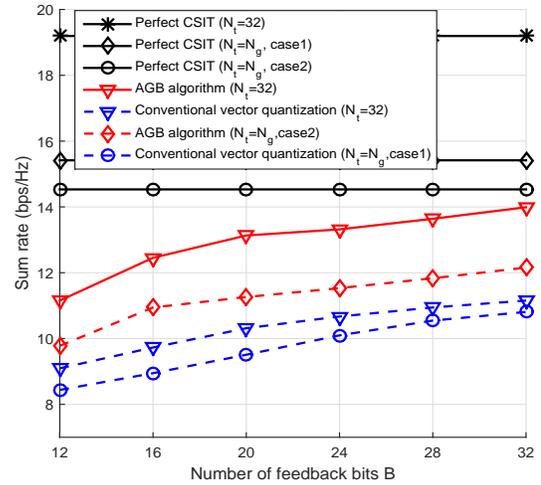, height = 70mm, width=80mm}}
\caption{Sum rate as a function of the number of feedback bits $B$ when SNR$=10$ dB and $B=(B_{q},8)$ ($N_{t}=32, M=4, K=4, N_{g}=16,\alpha=0.8$)} \label{fig:AGB_Fig9b}
\vspace{-0.2cm}
\end{figure}

We next measure the sum rate as a function of SNR. In this case, we set $N_{t}=32,M=4,K=4,\alpha=0.8$ and investigate the performance for two scenarios ($B=N_{t}$ and $2N_{t}$).
In addition, we plot the system with perfect CSIT as an upper bound.
As shown in Fig. \ref{fig:AGB_Fig9a}, the AGB algorithm achieves significant gain over the conventional vector quantization technique, bringing in more than 3 dB gain at mid SNR regime. 
In particular, with $B=2N_{t}$, the AGB algorithm performs within about 2 dB of perfect CSIT system until $\textrm{SNR}=7$ dB while others suffer from more than 5 dB loss compared to the perfect CSIT system.

\begin{figure}
 \centering
  \subfigure[]
  {\epsfig{figure=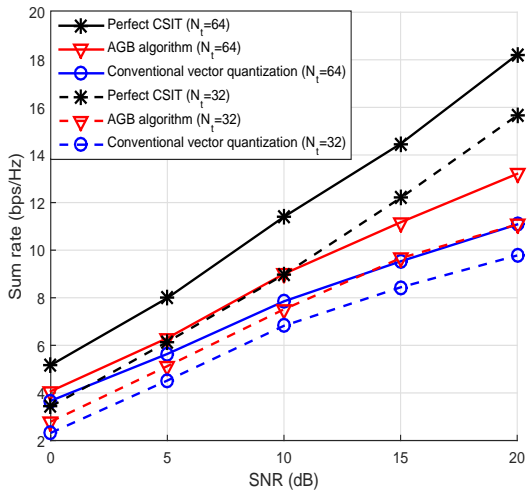, height = 70mm, width=80mm}}
  \vspace{-0.2cm}
   \subfigure[]
  {\epsfig{figure=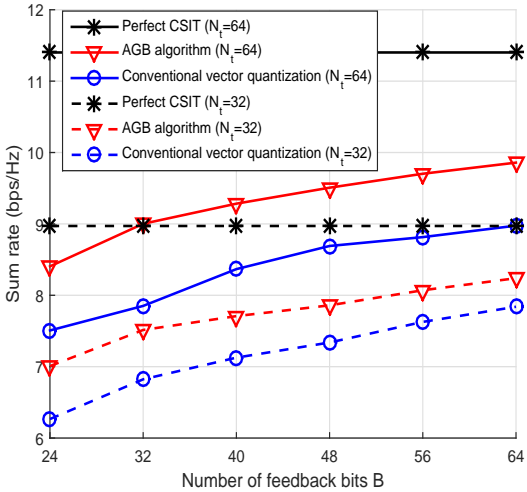, height = 70mm,  width=80mm}}
  \caption {Sum rate as a function of (a) SNR when $B=(24,8)$ for $N_{t}=32$ and $B=(48,16)$ for $N_{t}=64$ and (b) number of feedback bits $B$ for UPA correlation model when SNR$=10$ dB, $B=(B_{q},8)$ for $N_{t}=32$, and $B=(B_{q},16)$ for $N_{t}=64$.}
    \label{fig:AGB_Fig10}
    \vspace{-0.2cm}
\end{figure}

In Fig. \ref{fig:AGB_Fig9b}, we plot the sum rate as a function of the number of feedback bits. In this case, we set $K=4,\alpha=0.8,N_{g}=16$ and compare the performance of the AGB algorithm when $N_{t}=32$ with the following two scenarios; 1) a system having a reduced number of transmit antennas ($N_{t}=N_{g})$ and 2) a system where one antenna per group is selected and all remaining antennas per group is shut down ($N_{t}=N_{g}$). Interestingly, by taking advantage of high correlation among the antennas in a group, we observe that the performance of case 2) is better than that of the conventional system and case 1).
Nevertheless, due to the number of active antennas, the proposed algorithm when $N_{t}=32$ still achieves significant feedback overhead reduction over the case 2).
We observe that the proposed AGB algorithm with $N_{t}=32$ requires smaller number of bits to achieve the same level of performance. For example, the proposed approach achieves significant gain over the conventional vector quantization techniques, resulting in more than $60\%$ feedback overhead reduction.

\begin{table*}
\begin{center}
\caption{Simulation parameters for UPA model.}  
\begin{tabular}{l||l} \hline
 Variables & Simulation parameters \\ \hline
  Antenna elements spacing & $D=0.5$ \\
  Propagation path loss & $\gamma = \left(1+(\frac{s}{r})^{\alpha_{pl}}\right)^{-1}$ \\
  Path loss exponent & $\alpha_{pl} = 3$ \\
  Angular spread (vertical) & $\Delta_{V} = \frac{1}{2} \left( \arctan(\frac{s+r}{u}) - \arctan(\frac{s-r}{u})\right)$  \\
   Angle of arrival (vertical) & $\phi_{V} = \frac{1}{2}\left( \arctan(\frac{s+r}{u}) + \arctan(\frac{s-r}{u})\right)$  \\
   Angular spread (horizontal) & $\Delta_{H} = \arctan(\frac{r}{s})$ \\
   Angle of arrival (horizontal) & $\phi_{H,k} \in (-\pi,\pi]$ \\
   Elevation of the transmit antenna & $u=60m$ \\
   Radius of the scattering ring for the receiver & $r=30m$ \\
   Distance from the transmitter & $s=50m$ \\ \hline
  \end{tabular}
\end{center}
\label{tb:AGB_parameter}
\vspace{-0.2cm}
\end{table*}

In Fig. \ref{fig:AGB_Fig10}, we consider the two-dimensional UPA ($N_{\textrm{V}} \times N_{\textrm{H}}$ array) model, which is more realistic antenna model for massive MIMO scenarios.
The UPA model can be obtained by the Kronecker product of the vertical correlation matrix $\mathbf{R}_{V} \in \mathbb{C}^{N_{\textrm{V}} \times N_{\textrm{V}}}$ and the horizontal correlation matrix $\mathbf{R}_{H,k} \in \mathbb{C}^{N_{\textrm{H}} \times N_{\textrm{H}}}$.
 The resulting transmit correlation matrix of the UPA model is expressed as $\mathbf{R}_{t,k} = \mathbf{R}_{V} \otimes \mathbf{R}_{H,k}$ where $\otimes$ is the Kronecker product operator and each of the spatial correlation matrices is defined by
\begin{equation}
\label{eq:ULA_model}
[\mathbf{R}_{q,k}]_{m,p} = \frac{\gamma}{2\Delta_{q}} \int_{-\Delta_{q}+\phi_{q,k}}^{\Delta_{q}+\phi_{q,k}} e^{-j 2 \pi D (m-p) \sin(\alpha)} d\alpha
\end{equation}
where $q \in \{H,V\}$, $\gamma$ denotes propagation path loss between the transmitter and the receiver, $\Delta_{q}$ is the angular spread, $D$ is the antenna elements spacing, and $\phi_{q,k}$ is the angle of arrival (AoA) for the $k$-th user.
 We summarize the simulation parameters for UPA model in Table II.
In Fig. \ref{fig:AGB_Fig10}, we plot the sum rate as a function of SNR and the number of feedback bits for $N_{t} = 32,64$ and $K=2$.
For the UPA model, we set $N_{\textrm{V}}=4, N_{\textrm{H}}=8$ for $N_{t}=32,M=4$ and $N_{\textrm{V}}=8, N_{\textrm{H}}=8$ for $N_{t}=64,M=8$, respectively.
%
We observe from Fig. \ref{fig:AGB_Fig10}(a) that the proposed approach achieves better sum rate than the conventional scheme produces in particular for high SNR regime.
We also observe from Fig. \ref{fig:AGB_Fig10}(b) that the AGB algorithm outperforms the conventional vector quantization technique with a large margin, resulting in more than $50\%$ feedback overhead reduction.

So far, we have assumed that the receiver has knowledge of full CSI. In Fig. \ref{fig:AGB_MR_Fig1}, we investigate the performance of the AGB algorithm when the estimated CSI is employed. Since the mismatch between the actual CSI and the estimated CSI is unavoidable in a real communication, and this might result in degradation performance, it is of importance to investigate the effect of channel estimation error. In our simulation, we use an additive channel estimation model where $\mathbf{h}_{k,\textrm{est}} = \mathbf{h}_{k} + \mathbf{h}_{k,\textrm{err}}$ where $\mathbf{h}_{k,\textrm{est}},\mathbf{h}_{k}$ and $\mathbf{h}_{k,\textrm{err}}$ represent the estimated channel vector, the original channel vector and the estimated error vector, respectively. We assume that $\mathbf{h}_{k,\textrm{err}}$ is uncorrelated with $\mathbf{h}_{k,\textrm{est}}$, and $\mathbf{h}_{k,\textrm{err}}$ has i.i.d elements with zero mean and the estimation error variance $\sigma_{e,h}^{2}$. We observe from Fig. \ref{fig:AGB_MR_Fig1} that the AGB algorithm is more robust to the estimation errors than the conventional vector quantization. For example, the sum rate gain at $7$ bps/Hz of the AGB algorithm is about $5$ dB over the conventional vector quantization when $\sigma_{e,h}^{2}=0.05$, while the gain is around $3$ dB for $\sigma_{e,h}^{2}=0.01$.

\begin{figure}
\centerline{\psfig{figure=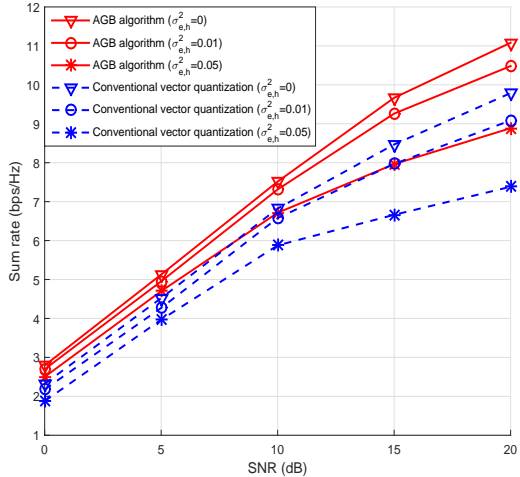, height = 70mm, width=80mm}} 
\caption{Sum rate performance with various $\sigma_{e,h}^{2}$ for UPA correlation model when $B=(24,8)$ $(N_{t}=32, M=4, K=4, N_{g}=16$).} \label{fig:AGB_MR_Fig1}
\vspace{-0.2cm}
\end{figure}

\begin{figure}
\centering
\centerline{\epsfig{figure=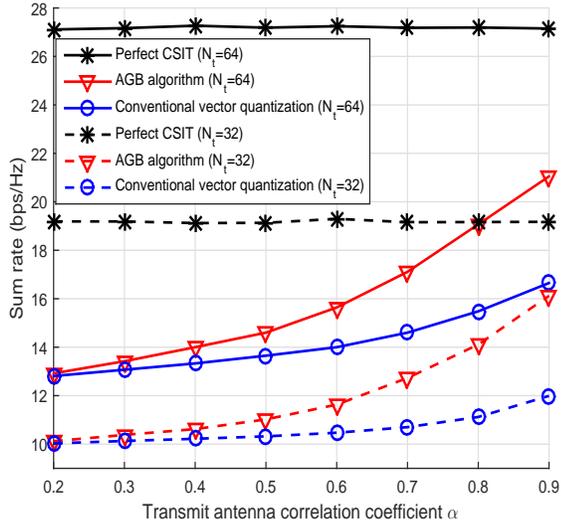, height = 75mm, width=85mm}}
\caption{Sum rate as a function of the transmit antenna correlation coefficient $\alpha$ when SNR$=10$ dB.}
\label{fig:AGB_Fig7}
\end{figure}

Finally, in Fig. \ref{fig:AGB_Fig7}, we plot the sum rate as a function of $\alpha$ for system with $N_{t} = 32,64$ and $K=4$.
In the AGB algorithm, we assign one bit per antenna elements on average. Specifically, for $N_{t}=32$, we set $B=(24,8)$ and for $N_{t}=64$, we set $B=(48,16)$, respectively. In this case, $N_{g}=16$ for $N_{t}=32$ and $N_{g}=32$ for $N_{t}=64$, respectively.
 It is worth mentioning that correlated fading tends to decrease the size of space that channel vectors span and hence is beneficial to reduce the quantization distortion by employing a user dependent channel statistic-based codebook \cite{Clerckx2}. As a result, the sum rate of multiuser MIMO systems increases with the transmit correlation coefficient.
 As shown in Fig. \ref{fig:AGB_Fig7}, when the transmit correlation coefficient $\alpha$ increases, the antenna grouping operation becomes effective and thus the sum rate of the AGB algorithm improves drastically. For example, when $\alpha = 0.8$, the sum rate gains of the AGB algorithm over the conventional vector quantization technique is $30\%$ for $N_{t}=32$ and $25\%$ for $N_{t}=64$, respectively.

\section{Conclusions}
\label{sec:AGB_secV}

In this paper, we proposed an efficient feedback reduction algorithm for FDD-based massive MIMO systems.
Our work is motivated by the observation that the CSI feedback overhead must scale linearly with the number of transmit antennas so that conventional vector quantization approach performing the quantization of the whole channel vector is not an appropriate option for the massive MIMO regime.
The key feature of the antenna group beamforming (AGB) algorithm is to control relentless growth of the CSI feedback information in the massive MIMO regime by mapping multiple correlated antenna elements into a single representative value with grouping patterns and then choosing the codeword from the codebook generated from the reduced dimension channel vector.
It has been shown by distortion analysis and simulation results that the proposed AGB algorithm is effective in achieving a substantial reduction in the feedback overhead in the realistic massive MIMO channels.

Although our study in this work focused on the single-cell scenario, we expect that the effectiveness of the proposed method can be readily extended to multi-cell scenario.
In fact, in the multi-cell scenario, more aggressive feedback compression is required since the channel information of the interfering cells as well as the desired cell may be needed at the basestation to properly control inter-cell interference.
In this scenario, the proposed AGB algorithm can be used as an effective means to achieve reduction in the feedback information.
Also, investigation of nonlinear transmitter techniques with user scheduling \cite{Byungju} would be interesting direction to be investigated.
Finally, we note that the proposed method can be nicely integrated into the dual codebooks structure in LTE-Advanced \cite{Dual,Lim} by feeding back the pattern index for long-term basis and the codebook index for short-term basis.
Since the main target of the massive MIMO system is slowly varying or static channels, dual codebook based AGB algorithm will bring further reduction in feedback overhead.

\appendices
\section{Derivation of (23)}

Denoting $(\bar{\mathbf{h}}_{\textrm{A}}^{H} \hat{\mathbf{h}}_{r})^{\ast}$ and $\bar{\mathbf{h}}_{\textrm{B}}^{H} \hat{\mathbf{h}}_{r}$ as $r_{1}\left(\cos\theta_{1}+j\sin\theta_{1}\right)$ and $r_{2}\left(\cos\theta_{2}+j\sin\theta_{2}\right)$, $\textrm{Re} [(\bar{\mathbf{h}}_{\textrm{A}}^{H} \hat{\mathbf{h}}_{r})^{\ast}(\bar{\mathbf{h}}_{\textrm{B}}^{H} \hat{\mathbf{h}}_{r})]$ becomes $r_{1}r_{2}(\cos\theta_{1}\cos\theta_{2}-\sin\theta_{1}\sin\theta_{2})$. Under the assumption that sufficient number of bits is used and hence the distortion between $\bar{\mathbf{h}}_{\textrm{A}}$ and $\hat{\mathbf{h}}_{r}$ is small (i.e, $\theta_{1} \approx 0, r_{1} \approx 1$), $E\left[\textrm{Re} \left[(\bar{\mathbf{h}}_{\textrm{A}}^{H} \hat{\mathbf{h}}_{r})^{\ast}(\bar{\mathbf{h}}_{\textrm{B}}^{H} \hat{\mathbf{h}}_{r})\right]\right]$ is expressed as
\begin{align}
&E\left[\textrm{Re} \left[(\bar{\mathbf{h}}_{\textrm{A}}^{H} \hat{\mathbf{h}}_{r})^{\ast}(\bar{\mathbf{h}}_{\textrm{B}}^{H} \hat{\mathbf{h}}_{r})\right]\right]  \nonumber \\
&= E\left[r_{1}r_{2}\left(\cos\theta_{1}\cos\theta_{2}-\sin\theta_{1}\sin\theta_{2}\right)\right] \nonumber \\
&= E\left[r_{1}r_{2}\right]E\left[\cos\theta_{1}\cos\theta_{2}-\sin\theta_{1}\sin\theta_{2}\right] \nonumber \\
&= E\left[r_{1}r_{2}\right]E\left[\cos(\theta_{1}+\theta_{2})\right] \nonumber \\
&\approx E\left[r_{2}\right]E\left[\cos\theta_{2}\right] \\
&= 0
\end{align}
where (39) follows from the fact that $\theta_{1} \approx 0, r_{1} \approx 1$, and (40) is because $E\left[\cos\theta_{2}\right]=0$ since $\theta_{2}$ is uniformly distributed between $-\pi$ and $\pi$.
\end{document}